\documentclass[11pt,letterpaper]{article}

\def\jamesmode{1}

\def\showauthornotes{0}
\def\showtableofcontents{1}
\def\showkeys{0}
\def\showdraftbox{0}

\usepackage{etex}

\usepackage[l2tabu, orthodox]{nag}

\usepackage{xspace,enumerate}

\usepackage[dvipsnames]{xcolor}

\usepackage[T1]{fontenc}
\usepackage[full]{textcomp}

\usepackage[american]{babel}

\usepackage{mathtools}

\usepackage{amsthm}

\newtheorem{theorem}{Theorem}[section]

\newtheorem{claim}[theorem]{Claim}

\newtheorem{lemma}[theorem]{Lemma}
\newtheorem{corollary}[theorem]{Corollary}
\newtheorem{conjecture}[theorem]{Conjecture}
\newtheorem{observation}[theorem]{Observation}
\newtheorem{fact}[theorem]{Fact}

\newtheorem{definition}[theorem]{Definition}

\newtheorem{problem}[theorem]{Problem}

\newtheorem{remark}[theorem]{Remark}

\usepackage[
letterpaper,
top=0.7in,
bottom=0.9in,
left=1in,
right=1in]{geometry}

\ifnum\jamesmode=0
\usepackage{newpxtext} 
\usepackage{textcomp} 
\usepackage[varg,bigdelims]{newpxmath}
\usepackage[scr=rsfso]{mathalfa}
\usepackage{bm} 
\let\mathbb\varmathbb
\fi

\ifnum\jamesmode=1
\usepackage{lmodern}
\usepackage{textcomp} 
\usepackage{amsmath,amssymb}
\usepackage[scr=rsfso]{mathalfa}
\usepackage{bm}
\fi

\ifnum\showkeys=1
\usepackage[color]{showkeys}
\fi

\usepackage[colorlinks,linkcolor=blue,filecolor=blue,citecolor=blue,urlcolor=blue]{hyperref}

\usepackage{prettyref}

\newrefformat{eq}{\hyperref[#1]{\textup{(\ref*{#1})}}}
\newrefformat{lem}{\hyperref[#1]{Lemma~\ref*{#1}}}
\newrefformat{def}{\hyperref[#1]{Definition~\ref*{#1}}}
\newrefformat{step}{\hyperref[#1]{step~\ref*{#1}}}
\newrefformat{fact}{\hyperref[#1]{Fact~\ref*{#1}}}
\newrefformat{thm}{\hyperref[#1]{Theorem~\ref*{#1}}}
\newrefformat{cor}{\hyperref[#1]{Corollary~\ref*{#1}}}
\newrefformat{conj}{\hyperref[#1]{Conjecture~\ref*{#1}}}
\newrefformat{prop}{\hyperref[#1]{Proposition~\ref*{#1}}}
\newrefformat{cha}{Chapter~\ref{#1}}
\newrefformat{sec}{\hyperref[#1]{Section~\ref*{#1}}}
\newrefformat{app}{\hyperref[#1]{Appendix~\ref*{#1}}}
\newrefformat{tab}{Table~\ref{#1}}
\newrefformat{fig}{\hyperref[#1]{Figure~\ref*{#1}}}
\newrefformat{item}{\hyperref[#1]{Item~\ref*{#1}}}

\newcommand{\Sref}[1]{\hyperref[#1]{\S\ref*{#1}}}

\usepackage{nicefrac}

\let\nfrac=\nicefrac

\usepackage{microtype}

\ifnum\showauthornotes=1
\newcommand{\Authornote}[2]{{\sf\small\color{red}{[#1: #2]}}}
\newcommand{\Authorcomment}[2]{{\sf \small\color{gray}{[#1: #2]}}}
\newcommand{\Authorfnote}[2]{\footnote{\color{red}{#1: #2}}}
\else
\newcommand{\Authornote}[2]{}
\newcommand{\Authorcomment}[2]{}
\newcommand{\Authorfnote}[2]{}
\fi

\usepackage{boxedminipage}

\newenvironment{mybox}
{\center \noindent\begin{boxedminipage}{1.0\linewidth}}
{\end{boxedminipage}
\noindent
}

\newcommand{\abs}[1]{\lvert#1\rvert}

\newcommand{\set}[1]{\{#1\}}
\newcommand{\Set}[1]{\left\{#1\right\}}

\newcommand{\norm}[1]{\lVert#1\rVert}

\newcommand{\iprod}[1]{\langle#1\rangle}

\newcommand{\Esymb}{\mathbb{E}}
\newcommand{\Psymb}{\mathbb{P}}

\DeclareMathOperator*{\E}{\Esymb}
\DeclareMathOperator*{\Var}{\mathbf{Var}}
\DeclareMathOperator*{\Cov}{\mathbf{Cov}}
\DeclareMathOperator*{\ProbOp}{\Psymb}

\renewcommand{\Pr}{\ProbOp}

\newcommand{\suchthat}{\;\middle\vert\;}

\ifx\because\undefined
\newcommand{\because}[1]{\text{(because #1)}}
\else
\renewcommand{\because}[1]{\text{(because #1)}}
\fi

\newcommand{\maximize}{\mathop{\textrm{maximize}}}

\newcommand{\subjectto}{\mathop{\textrm{subject to}}}

\newcommand{\bits}{\{0,1\}}

\newcommand{\e}{\epsilon}

\newcommand{\sse}{\subseteq}

\newcommand{\vbig}{\vphantom{\bigoplus}}

\newcommand{\defeq}{\stackrel{\mathrm{def}}=}     

\renewcommand{\vec}[1]{{\bm{#1}}}

\newcommand{\mper}{\,.}
\newcommand{\mcom}{\,,}

\newcommand\bdot\bullet

\ifx\mathds\undefined 
\newcommand{\Ind}{\mathbb I}
\else
\newcommand{\Ind}{\mathds 1}
\fi

\newcommand{\etal}{et al.\xspace}

\DeclareMathOperator{\Inf}{Inf}

\DeclareMathOperator{\val}{val}

\DeclareMathOperator{\poly}{poly}

\DeclareMathOperator{\supp}{supp} 
 
\DeclareMathOperator{\sign}{sign}

\newcommand{\Z}{\mathbb Z}
\newcommand{\N}{\mathbb N}
\newcommand{\R}{\mathbb R}

\newcommand{\prb}[1]{\problemmacro{#1}}

\newcommand{\problemmacro}[1]{\texorpdfstring{\textsc{#1}}{#1}}

\newcommand{\cA}{\mathcal A}
\newcommand{\cB}{\mathcal B}
\newcommand{\cC}{\mathcal C}

\newcommand{\cF}{\mathcal F}
\newcommand{\cG}{\mathcal G}

\newcommand{\cL}{\mathcal L}

\newcommand{\cO}{\mathcal O}
\newcommand{\cP}{\mathcal P}

\newcommand{\cR}{\mathcal R}

\newcommand{\cV}{\mathcal V}

\newcommand{\ssimp}{{\displaystyle\blacktriangle}}

\renewcommand{\leq}{\leqslant}
\renewcommand{\le}{\leqslant}
\renewcommand{\geq}{\geqslant}
\renewcommand{\ge}{\geqslant}

\ifnum\showdraftbox=1
\newcommand{\draftbox}{\begin{center}
  \fbox{%
    \begin{minipage}{2in}%
      \begin{center}%
        \begin{Large}%
          \textsc{Working Draft}%
        \end{Large}\\
        Please do not distribute%
      \end{center}%
    \end{minipage}%
  }%
\end{center}
\vspace{0.2cm}}
\else
\newcommand{\draftbox}{}
\fi

\let\epsilon=\varepsilon
\let\eps=\varepsilon

\numberwithin{equation}{section}

\let\origparagraph\paragraph
\renewcommand{\paragraph}[1]{\origparagraph{#1.}}

\newcommand{\DSstore}[2]{%
  \global\expandafter \def \csname DSMEMORY #1 \endcsname{#2}
}

\newcommand{\DSload}[1]{%
  \csname DSMEMORY #1 \endcsname%
}

\newcommand{\DSnewlabel}[1]{%
  \newcommand\DScurrentlabel{#1}%
  \DSoldlabel{#1}%
}

\newcommand{\DSdummylabel}[1]{}

\newcommand{\torestate}[1]{%
  \let\DSoldlabel\label%
  \let\label\DSnewlabel%
  #1%
  \DSstore{\DScurrentlabel}{#1}%
  \let\label\DSoldlabel%
}

\newcommand{\restatetheorem}[1]{%
  \let\DSoldlabel\label
  \let\label\DSdummylabel
  \begin{theorem*}[Restatement of \prettyref{#1}]
    \DSload{#1}
  \end{theorem*}
  \let\label\DSoldlabel
}

\newcommand{\restatelemma}[1]{%
  \let\DSoldlabel\label
  \let\label\DSdummylabel
  \begin{lemma*}[Restatement of \prettyref{#1}]
    \DSload{#1}
  \end{lemma*}
  \let\label\DSoldlabel
}

\newcommand{\restateprop}[1]{%
  \let\DSoldlabel\label
  \let\label\DSdummylabel
  \begin{proposition*}[Restatement of \prettyref{#1}]
    \DSload{#1}
  \end{proposition*}
  \let\label\DSoldlabel
}

\usepackage{makeidx}

\makeindex

\newcommand{\KnnQuadraticProgramming}%
{\textsc{$K_{N,N}$-QuadraticProgramming}\xspace}
\newcommand{\KmnQuadraticProgramming}%
{\textsc{$K_{M,N}$-QuadraticProgramming}\xspace}

\renewcommand{\P}{\mathrm{P}}
\newcommand{\NP}{$\mathrm{NP}$}

\newcommand{\inst}{\Im}

\renewcommand{\NP}{$\mathsf{NP}$}

\newcommand{\round}{\msf{Round}}

\newcommand{\qV}[1]{\bm{b}_{#1}}

\newcommand{\vzero}{\mathbf{I}}

\newcommand{\ssimptrunc}{f_{\ssimp}}

\newcommand{\msf}[1]{\mathsf{#1}}
\newcommand{\mcl}[1]{\mathcal{#1}}

\newcommand{\mv}[1]{\mathbf{#1}}

\newcommand{\mrv}[1]{\bm{\lowercase{#1}}}

\newcommand{\erv}[1]{\mathcal{#1}}

\newcommand{\mpl}[1]{\bm{\uppercase{#1}}}

\newcommand{\lsat}{\prb{$\Lambda$-Sat}\xspace}
\newcommand{\maxl}{\prb{Max-$\Lambda$}\xspace}

\let\pref=\prettyref

\newcommand{\eat}[1]{}

\DeclareMathOperator{\Lin}{Lin}
\DeclareMathOperator{\Sp}{\mathsf{Span}}

\newcommand{\uglbl}{R}
\newcommand{\polymorph} {\mathrm{Poly}}

\title{Combinatorial Optimization Algorithms via Polymorphisms}

\author{%
		Jonah Brown-Cohen \\
		{University of California, Berkeley}\\
{California, Berkeley}\\
 \texttt{jonahbc@eecs.berkeley.edu}
\and
Prasad Raghavendra \thanks{Supported by NSF Career Award and Alfred.
	P. Sloan
Fellowship}\\
{University of California, Berkeley}\\
{California, Berkeley}\\
 \texttt{prasad@eecs.berkeley.edu}
}

\newif\iffull

\fulltrue

\begin{document}

\maketitle
\draftbox
\thispagestyle{empty}

\begin{abstract}
	
An elegant characterization of the complexity of constraint
satisfaction problems has emerged in the form of the the algebraic dichotomy
conjecture of \cite{BulatovKJ00}.  Roughly
speaking, the characterization asserts that a CSP $\Lambda$ is tractable
if and only if there exist certain non-trivial operations known as
{\it polymorphisms} to combine
solutions to $\Lambda$ to create new ones.  In an entirely
separate line of work, the unique games conjecture yields a
characterization of approximability of Max-CSPs.  Surprisingly, this
characterization for Max-CSPs can also be reformulated in the language of polymorphisms.

In this work, we study whether existence of non-trivial polymorphisms implies
tractability beyond the realm of constraint satisfaction problems,
namely in the value-oracle model.  Specifically, given a function $f$
in the value-oracle model along with an appropriate operation that never increases
the value of $f$, we design algorithms to minimize $f$.  In
particular, we design a randomized algorithm to minimize a function
$f:[q]^n \to \R$ that admits a fractional
polymorphism which is {\it measure preserving} and has a {\it
transitive symmetry}.

We also reinterpret known results on MaxCSPs and thereby reformulate the unique games conjecture as a characterization
of approximability of max-CSPs in terms of their {\it approximate
polymorphisms}.

\end{abstract}

\clearpage

\ifnum\showtableofcontents=1
{
\tableofcontents
\thispagestyle{empty}
 }
\fi

\clearpage
\setcounter{page}{1}
\section{Introduction}
%

A vast majority of natural computational problems have been classified
to be either polynomial-time solvable or \NP-complete.   While there is
little progress in determining the exact time complexity for
fundamental problems like matrix multiplication, it
can be argued that a much coarser classification of P vs
\NP-complete has been achieved for a large variety of
problems.  Notable problems that elude such a classification
include factorization or graph isomorphism.

A compelling research direction at this juncture is to
understand what causes problems to be easy (in P) or
hard (\NP-complete).  More precisely, for specific classes
of problems, does there exist a unifying theory that explains
and characterizes why some problems in the class are in P
while others are \NP-complete?  For the sake of concreteness,
we will present a few examples.

It is well-known that $\prb{2-Sat}$ is polynomial-time
solvable, while $\prb{3-Sat}$ is
\NP-complete.  However, the traditional proofs of these
statements are unrelated to each other and therefore shed
little light on what makes $\prb{2-Sat}$ easy while
$\prb{3-Sat}$ \NP-complete.  Similar situations arise even
when we consider approximations for combinatorial optimization problems.  For instance, $\prb{Min
  s-t Cut}$ can be solved exactly, while $\prb{Min
  3-way cut}$ is \NP-complete and can be approximated to a factor
of $\frac{12}{11}$.  It is natural to ask why the approximation ratio for $\prb{Min s-t Cut}$ is $1$,
while it is $\frac{12}{11}$ for $\prb{Min 3-way cut}$.

Over the last decade, unifying theories of tractability have
emerged for the class of constraint satisfaction problems
(CSP) independently for the exact and the optimization
variants.  Surprisingly, the two emerging theories for the
exact and optimization variants of CSPs appear to coincide!
While these candidate theories remain conjectural for now, they
successfully explain all the existing
algorithms and hardness results for CSPs.  To set the stage
for the results of this paper, we begin with a brief survey
of the theory for CSPs.
\vspace{-3mm}
\paragraph{Satisfiability}
A constraint satisfaction problem (CSP) $\Lambda$ is specified
by a family of predicates over a finite domain $[q] =
\{1,2,\ldots,q\}$. Every instance of the CSP $\Lambda$
consists of a   set of variables $\cV$, along with a set of
constraints $\cC$ on them.
Each constraint in $\cC$ consists of a predicate from the
family $\Lambda$ applied to a subset of variables.  For a CSP
$\Lambda$, the associated satisfiability problem
\lsat is defined as follows.
\begin{problem}[\lsat] Given an instance $\inst$ of
  the CSP $\Lambda$, determine whether there is an assignment satisfying
  all the constraints in $\inst$.
\end{problem}
A classic theorem of Schaefer \cite{Schaefer78} asserts that
among all satisfiability problems  over the boolean domain ($\{0,1\})$),
only \prb{Linear-Equations-Mod-2}, \prb{2-Sat},
\prb{Horn-Sat}, \prb{Dual-Horn Sat} and certain trivial CSPs are solvable in polynomial
time.  The rest of the boolean CSPs are \NP-complete.
The dichotomy conjecture of Feder and Vardi \cite{FederV98}
asserts that every \lsat is in P or \NP-complete.  The
conjecture has been shown to hold for CSPs over domains of size up to
$4$ \cite{Bulatov06}.
%

In this context, it is natural to question as to what makes
certain \lsat problems tractable while the others are \NP-complete.
Bulatov, Jeavons and Krokhin \cite{BulatovKJ00} conjectured
a beautiful characterization for tractable satisfiability problems.  We
will present an informal description of this characterization
known as the { \it algebraic
  dichotomy conjecture}.  We refer the reader to the work of Kun Szegedy
\cite{KunS09} for a more formal description.

To motivate this characterization, let us consider
a CSP known as the \prb{XOR} problem.  An instance of
the \prb{XOR} problem consists of a
system of linear equations over $\Z_2 = \{0,1\}$.  Fix an instance $\inst$ of
$\prb{XOR}$ over $n$ variables.  Given three solutions $X^{(1)},
X^{(2)}, X^{(3)} \in \bits^n$ to $\inst$, one can create a new
solution  $Y \in \bits^n$:
$$ Y_i = X^{(1)}_i \oplus X^{(2)}_i \oplus X^{(3)}_i \qquad \forall i \in
[n] \mper$$
It is easy to check that $Y$ is also a feasible solution to the
instance $\inst$.  Thus the $XOR: \bits^3 \to \bits$ yields a way to
combine three solutions in to a new solution for the same instance.
Note that the function $XOR$ was applied to each bit of the solution
individually.
An operation of this form that preserves the satisfiability of the
CSP is known as a {\em polymorphism}.  Formally, a
polymorphism of a CSP \lsat is defined as follows:
\begin{definition}[Polymorphisms]
  A function $p : [q]^\uglbl \to [q]$ is said to be {\it a polymorphism}
  for the CSP \lsat, if for every instance $\inst$ of $\Lambda$ and
  $\uglbl$ assignments $X^{(1)},X^{(2)}, \ldots, X^{(\uglbl)} \in [q]^{n}$
  that satisfy all constraints in $\inst$, the vector $Y \in [q]^n$
  defined below is also a feasible solution.
  $$ Y_i = p(X^{(1)}_i, X^{(2)}_i, X^{(3)}_i,\dots, X_i^{(\uglbl)}) \qquad \forall i \in
  [n] \mper$$
\end{definition}

Note that the dictator functions $p(x^{(1)},\ldots,x^{(R)} ) = x^{(i)}$ are polymorphisms
for every CSP \lsat.  These will be referred to as {\it
  projections} or trivial polymorphisms.
All the tractable cases of boolean CSPs in Schaefer's theorem are
characterized by existence of non-trivial polymorphisms.  Specifically, \prb{2-SAT} has the Majority
functions, \prb{Horn-SAT} has OR functions, and \prb{Dual Horn-SAT}
has the AND functions as polymorphisms.  Roughly speaking, Bulatov
\etal \cite{BulatovKJ00} conjectured that the
existence of non-dictator polymorphisms characterizes CSPs that are tractable.
Their work showed that the set of polymorphisms $\polymorph(\Lambda)$
of a CSP $\Lambda$ characterizes the
complexity of \lsat.
\iffull
Formally,
\begin{theorem}\cite{BulatovKJ00}
  If two CSPs $\Lambda_1$, $\Lambda_2$ have the same set of
  polymorphisms $\polymorph(\Lambda_1) = \polymorph(\Lambda_2)$, then
  \prb{$\Lambda_1$-sat} is polynomial-time reducible to \prb{$\Lambda_2$-sat}
  and vice versa.
\end{theorem}
\fi
There are many equivalent ways of formalizing what it means for an operation
to be {\it non-trivial} or {\it non-dictator}.  A particularly simple
way to formulate the algebraic dichotomy conjecture arises out of the
recent work of Barto and Kozik \cite{BartoK12}.  A polymorphism $p:
[q]^k \to [q]$ is called a {\it cyclic term} if
$$ p(x_1, \ldots, x_k) = p(x_2, \ldots, x_k, x_1) = \ldots = p(x_k,
x_{k-1},\ldots, x_1) \quad \forall x_1,\ldots,x_k \in [q] \mper $$
Note that the above condition strictly precludes the operation $p$
from being a dictator.
\begin{conjecture} ( \cite{BulatovKJ00,MarotiMMR08,BartoK12})
  \lsat is in $\P$ if $\Lambda$ admits a cyclic
  term, otherwise \lsat is \NP-complete.
\end{conjecture}
Surprisingly, one of the implications of the conjecture has already been
resolved.
\begin{theorem} \cite{BulatovKJ00, MarotiMMR08,BartoK12}
  \lsat is \NP-complete if $\Lambda$ does not
  admit a cyclic term.
\end{theorem}

The algebraic approach to understanding the complexity of CSPs
has received much attention, and the algebraic dichotomy conjecture
has been verified for several subclasses of CSPs such as conservative
CSPs \cite{Bulatov03}, CSPs with no ability to count \cite{BartoK09}
and CSPs with Maltsev operations \cite{BulatovD06}.  Recently, Kun and Szegedy reformulated the algebraic dichotomy
conjecture using analytic notions similar to influences \cite{KunS09}.

%
%
%


\vspace{-3mm}
\paragraph{Polymorphisms in Optimization}
Akin to polymorphisms for \lsat, we define the
notion of {\it approximate polymorphisms} for optimization
problems.
Roughly speaking, an approximate polymorphism is a probability
distribution $\cP$ over a set of operations of the form $p:[q]^k \to
[q]$.  In particular, an approximate polymorphism $\cP$ can be
used to combine $k$ solutions to an optimization problem
and produce a probability distribution over new solutions to
the same instance. Unlike the case of exact CSPs, here
the polymorphism outputs several new solutions.  $\cP$ is an
$(c,s)$-approximate polymorphism if, given a set of solutions
with average value $c$, the average value of the
output solutions is always at least $s$.

For the sake of exposition, we present an example here.
Suppose $f : \bits^n \to \R^+$ is a {\it super-modular} function.  In this case,
$f$ admits a $1$-approximate polymorphism described below.
$$\cP = \begin{cases} OR(x_1,x_2) & \text{ with probability }
  \frac{1}{2} \\
  AND(x_1,x_2) & \text{ with probability }
  \frac{1}{2} \end{cases} $$
Given any two assignments $X^{(1)}$, $X^{(2)}$, the polymorphism $\cP$ outputs two solutions
$X^{(1)} \vee X^{(2)}$ and $X^{(1)} \wedge X^{(2)}$.  The
supermodularity of the function $f$ implies
$$ \val(X^{(1)}) + \val(X^{(2)}) \leq \val(X^{(1)} \vee
X^{(2)}) + \val(X^{(1)} \wedge X^{(2)})\ , $$
that the average value of solutions output by $\cP$ is at
least the average value of solutions input to it.
The formal definition of an $(c,s)$-approximate
polymorphism is as follows.
\begin{definition}[$(c,s)$-approximate polymorphism]
  Fix a function $f: [q]^n \to \R^+$.
  A probability distribution $\cP$ over operations $p:
  [q]^{\uglbl} \to [q]$ is an $(c,s)$-approximate polymorphism
  for $f$ if the
  following holds:
  for every $\uglbl$ assignments $X^{(1)},
  X^{(2)},\dots, X^{(\uglbl)} \in [q]^n$, if $\E_i f(X^{(i)}) \geq c$ for
  all $i$ then,
  $$ \E_{p \in \cP} [f(p(X^{(1)},\dots,X^{(\uglbl)}))] \geq s $$
  Here, $p(X^{(1)},\dots,X^{(\uglbl)})$ is the assignment obtained by
  applying the operation $p$ coordinate-wise.
  %
\end{definition}
It is often convenient to use a coarser notion of approximation,
namely the approximation ratio for all $c$.  In this vein, we define
$\alpha$-approximate polymorphisms.
\begin{definition}[$\alpha$-approximate polymorphism]
  A probability distribution $\cP$ over operations $p:
  [q]^{\uglbl} \to [q]$ is an $\alpha$-approximate polymorphism
  for $f : [q]^n \to \R^+$, if it is an
  $(c, \alpha \cdot c)$-approximate polymorphism for all $c \geq
  0$.
\end{definition}
We will refer to a $1$-approximate polymorphism as a {\it
  fractional polymorphism} along the lines of Cohen \etal
\cite{CohenCJ06}.

\paragraph{Approximating CSPs}
To set the stage for the main result of the paper, we will
reformulate prior work on approximation of CSPs in the language of
polymorphisms.  While the results stated here are not new,
their formulation in the language of approximate polymorphisms
has not been carried out earlier.  This reformulation highlights the the close connection between
the algebraic dichotomy conjecture (a theory for exact CSPs) and the unique games
conjecture (a theory of approximability of CSPs).  We believe
this reformulation is also one of the contributions of this
work and have devoted
\iffull
\pref{sec:maxcsps}
\fi
to give a full account of it.
The results we will state now hold for a fairly general class of
problems that include both maximization and minimization
problems with local bounded payoff functions (referred to as
{\it value CSP} in \cite{CohenCCJZ12}).  However,  for
the sake of concreteness, let us restrict our attention to the problem
\maxl for a CSP $\Lambda$.
\begin{problem}[\maxl] Given an instance $\inst$ of
  the $\Lambda$-CSP, find an assignment that satisfies the maximum
  number (equivalently fraction) of constraints.
\end{problem}

The formal definition of $\alpha$-approximate polymorphisms is
as follows.
\begin{definition} \label{def:approxpolycsps}
  A probability distribution $\cP$ over operations
  $p:[q]^k \to [q]$ is an $\alpha$-approximate
  polymorphism for \maxl, if $\cP$ is an $(c,\alpha \cdot c)$-approximate
  polymorphism for every instance $\inst$ of \maxl.
\end{definition}
In the above definition, we treat each instance $\inst$ as a
function $\inst: [q]^n \to \R^+$ that we intend to maximize.
As in the case of \lsat, dictator functions are $1$-approximate polymorphisms for every
\maxl.  We will use the analytic notion of
{\it influences} to preclude the dictator functions (see
\iffull
\pref{sec:maxcsps}
\fi).
Specifically, we define $\tau$-quasirandomness as follows.
\begin{definition}
  An approximate polymorphism $\cP$ is $(\tau,d)$-quasirandom if, for
  every distribution $\mu$ on $[q]$,
  $$ \E_{p \in \cP} \left[\max_i \Inf_{i,\mu}^{< d}(p)\right] \leq \tau \quad \forall
  i \in [k]$$
\end{definition}

By analogy to \lsat, it is natural to conjecture that
approximate polymorphisms characterize the approximability of every
\maxl.  Indeed, all known tight approximability results
for different Max-CSPs like \prb{3-Sat} correspond exactly to the
existence of approximate polymorphisms.  Moreover, we show that the
unique games conjecture \cite{Khot02a} is equivalent
to asserting that existence of quasirandom approximate polymorphisms
characterizes the complexity of approximating every CSP.  To state our
results formally, let us define some notation.  For a CSP $\Lambda$
define the approximation constant $\alpha_{\Lambda}$ as,
$$ \alpha_{\Lambda} \defeq \sup \left\{ \alpha \in \R \mid \forall \tau,d > 0 \ \  \exists
  (\tau,d)\text{-quasirandom, } \alpha\text{-approximate polymorphism
    for } \Lambda\right\} \mper$$
Similarly, for each $c > 0$ define the constant $s_{\Lambda}(c)$
as,
$$ s_{\Lambda}(c) \defeq \sup \left\{ s \in \R \mid \forall \tau,d > 0 \ \  \exists
  (\tau,d)\text{-quasirandom, } (c,s)\text{-approximate polymorphism
    for } \Lambda\right\} \mper$$
%
%
%
%
The connection between unique games conjecture (stated in
\iffull
\pref{sec:maxcsps}
\fi) and approximate
polymorphisms is summarized below.
\begin{theorem}[Connection to UGC] \label{thm:ughardness}
  Assuming the unique games conjecture, for every $\Lambda$,
  it is \NP-hard to approximate \maxl better than
  $\alpha_{\Lambda}$.  Moreover, the unique games conjecture is equivalent to the following
  statement: For every $\Lambda$ and
  $c$,  on instances of $\maxl$ with value $c$, it is \NP-hard to find an
  assignment with value larger than $s_{\Lambda}(c)$.
\end{theorem}
All unique games based hardness results for \maxl are
shown via constructing appropriate dictatorship testing gadgets
\cite{KhotKMO07,Raghavendra08}.
The above theorem is a consequence of the connection between
approximate polymorphisms and dictatorship testing (see
\iffull
\prettyref{sec:approxpoly}
\fi
for details).
The connection between
approximate polymorphisms and dictatorship testing was observed
in the work of Kun and Szegedy \cite{KunS09}.

Surprisingly, the other direction of the correspondence between
approixmate polymorphisms and complexity of CSPs also holds.
Formally, we show the following result about a semidefinite
programming relaxation for CSPs referred to as the {\sc Basic SDP}
relaxation (see \pref{sec:roundingcsps} for definition).
\begin{theorem}[Algorithm via Polymorphisms] \label{thm:sdpgap}
  For every CSP $\Lambda$, the integrality gap of the {\sc
    Basic-SDP} relaxation for \maxl is at
  most $\alpha_{\Lambda}$.  More precisely, for every instance $\inst$ of \maxl,
  for every $c$, if the optimum value of the {\sc Basic SDP} relaxation is at least
  $c$, then optimum value for $\inst$ is at least $\lim_{\epsilon \to
    0} s_{\Lambda}(c-\e)$.
\end{theorem}

In particular, the above theorem implies that the {\sc Basic-SDP}
relaxation yields an $\alpha_{\Lambda}$ (or $s_{\Lambda}(c)$)
approximation to \maxl.  This theorem is a consequence
of restating the soundness analysis of \cite{Raghavendra08} in the
language of approximate polymorphisms.  Specifically, one uses the
$\alpha$-approximate polymorphism to construct an $\alpha$-factor rounding algorithm
to the semidefinite program for \maxl (See
\iffull
\prettyref{sec:roundingcsps}
\fi
for details).
\iffull
In the case of the algebraic dichotomy conjecture, the \NP-hardness
result is known, while an efficient algorithm for all CSPs via
polymorphisms is not known.  In contrast, the {\sc Basic SDP} relaxation
yields an efficient algorithm for all \maxl, but the
\NP-hardness for max-CSPs is open and equivalent to the unique
games conjecture.
\fi

The case of \maxl that are solvable
exactly (approximation ratio = $1$) has also received considerable
attention in the literature.
Generalising the algebraic approach to CSPs, algebraic
properties called multimorphisms \cite{CohenCJK06}, fractional
polymorphisms \cite{CohenCJ06} and
weighted polymorphisms \cite{CohenCCJZ12} have been proposed to study the complexity
of classes of valued CSPs.  The notion of $\alpha$-approximate polymorphism for $\alpha = 1$
is closely related to these notions (see
\iffull
\pref{sec:cspdefs}
\fi
).

In a recent work, Thapper and Zivny \cite{ThapperZ12}
obtained a characterization of \maxl that can be solved
exactly.  Specifically, they showed that a valued CSP (a
generalization of maxCSP) is tractable if and only if it admits a
symmetric $1$-approximate polymorphism on two inputs.
This is further evidence supporting the
claim that approximate polymorphisms characterize approximability of
max-CSPs.

\iffull
\paragraph{Beyond CSPs}

It is natural to wonder if a similar theory of tractability could be developed
for classes of problems beyond constraint satisfaction problems.  For
instance, it would be interesting to understand the tractability of
\prb{Minimum Spanning Tree}, \prb{Minimum Matching}, the
intractability of \prb{TSP} \prb{Steiner Tree} and
the approximability of \prb{Steiner Tree}, \prb{Metric TSP} and
\prb{Network Design} problems.

It appears that tractable problems such as \prb{Minimum Spanning Tree}
and \prb{Maximum Matching} have certain ``operations'' on the solution
space.  For instance, the union of two perfect matchings contains
alternating cycles that could be used to create new matchings.
Similarly, spanning trees form a matroid and therefore given two
spanning trees, one could create a sequence of trees by exchanging
edges amongst them.  Furthermore, some algorithms for these problems crucially exploit
these operations.  More indirectly, \prb{Minimum Spanning Tree} and
\prb{Maximum Matching} are connected to submodularity via
polymatroids.  As we saw earlier, submodularity is an example of a
$1$-approximate polymorphism.

Moreover, here is a heuristic justification for the connection between
tractability and existence of operations to combine solutions.
Typically, a combinatorial optimization problem like \prb{Maximum
  Matching} is tractable because of the existence of a linear
programming relaxation $P$ and an associated rounding scheme $\cR$
(possibly randomized).  Given
a set of $k$ integral solutions $X^{(1)},X^{(2)},\ldots,X^{(k)}$,
consider any point in their convex hull, say $y = \frac{1}{k} \sum_{i
  \in [k]} X^{(i)}$.  By convexity, the point $y$ is also a feasible
solution to the linear programming relaxation $P$ .  Therefore, we
could execute the rounding scheme $\cR$ on the LP solution $y$ to
obtain a distribution over integral solutions.
Intuitively, $y$ has less information than $X^{(1)},\ldots,X^{(k)}$
and therefore the solutions output by the rounding scheme $\cR$ should
be different from $X^{(1)}, \ldots,X^{(k)}$.  This suggests that the
rounding scheme $\cR$ yields an operation to combine
solutions to create new solutions.  Recall that in case of CSPs, indeed polymorphisms are used to obtain rounding
schemes for the semidefinite programs (see \prettyref{thm:sdpgap}).
More recently, Barak, Steurer and Kelner \cite{BarakKS14} show
how certain algorithms to {\it combine} solutions can be used towards
rounding sum of squares relaxations.
\fi

\subsection{Our Results}

Algorithmically, one of the fundamental results in combinatorial
optimization is polynomial-time minimization of arbitrary submodular
functions.  Specifically, there exists an efficient algorithm to
minimize a submodular function $f$ given access to a value oracle for
$f$ \cite{Cunningham81,Cunningham83,Cunningham84,Schrijver00}.
Since submodularity is an example of a
fractional polymorphism, it is natural to conjecture that such an
algorithm exists whenever $f$ admits a certain fractional polymorphism.
Our main result is an efficient algorithm to minimize a function $f$
given access to its value oracle, provided $f$ admits an appropriate
fractional polymorphism.  Towards stating our result formally,
let us define some notation.

\begin{definition}
  A fractional polymorphism is said to be {\it measure-preserving} if for each $i
  \in [q]$ and for every choice of inputs,  the fraction
  of inputs equal to $i$ is always equal to the fraction of outputs
  equal to $i$.
\end{definition}
\begin{definition}
  An operation $p: [q]^k \to [q]$ is said to have a {\it
    transitive symmetry} if for all $i, j \in [k]$ there exists a
  permutation $\sigma_{ij} \in S_{k}$ of inputs such that
  $\sigma_{ij}(i) = j$ and $p \circ \sigma = p$.  A polymorphism $\cP$
  is said to be {\it transitive symmetric} if every operation $p \in
  \supp(\cP)$ is transitive symmetric.
\end{definition}
Notice that cyclic symmetry of the operation is a special case of
transitive symmetry.  As per our definitions, submodularity is a {\it measure-preserving}
and {\it transitive symmetric} polymorphism.

\begin{theorem} \label{thm:main}
  Let $f: [q]^n \to \R^+$ be a function that admits a fractional
  polymorphism $\cP$.  If $\cP$ is measure preserving and transitive
  symmetric then there exists an efficient randomized
  algorithm $\cA$ that for every $\epsilon > 0$, makes $\poly(1/\epsilon, n)$ queries to
  the values of $f$ and outputs an $x \in [q]^n$ such that,
  $$ f(x) \leq \min_{y \in [q]^n} f(y) + \epsilon
  \norm{f}_{\infty} $$
\end{theorem}
Apart from submodularity, here is an example of a {\it measure
  preserving} and {\it transitive symmetric} polymorphism.
$$ \cP = \begin{cases} \mathrm{Majority}(x_1,x_2,x_3) & \text{ with
    probability } \nfrac{2}{3} \\ \mathrm{Minority}(x_1,x_2,x_3) &
  \text{with probability } \nfrac{1}{3} \end{cases} $$
Here $\mathrm{Minority}(0,0,0) = 0, \mathrm{Minority}(1,1,1) =
1$ and $\mathrm{Minority}(x_1,x_2,x_3) = 1-
\mathrm{Majority}(x_1,x_2,x_3)$ on the rest of the inputs.
On a larger domain $[q]$, a natural example of a fractional
polymorphism would be
$$ \cP = \begin{cases} \mathrm{max}(x_1,x_2,x_3) & \text{ with
    probability } \nfrac{1}{3} \\ \mathrm{min}(x_1,x_2,x_3) &
  \text{with probability } \nfrac{1}{3} \\
  \mathrm{median}(x_1,x_2,x_3) &
  \text{with probability } \nfrac{1}{3}
\end{cases}
$$
More generally, it is very easy to construct examples of fractional
polymorphisms that satisfy the hypothesis of \pref{thm:main}.

\subsection{Related Work}

Operations that give rise to tractability in the value oracle
model have received
considerable attention in the literature.  A particularly
successful line of work studies generalizations of submodularity over various lattices.
In fact, submodular functions given by an oracle can be
minimised on distributive lattices
\cite{IwataFF01,Schrijver00}, diamonds \cite{Kuivinen11}, and the
pentagon \cite{KrokhinL08} but the case of general non-distributive
lattices remains open.

An alternate generalization of submodularity known as {\it
  bisubmodular functions} introduced by Chandrasekaran and
Kabadi \cite{ChandrasekaranK88} arises naturally both in
theoretical and practical contexts (see
\cite{FujishigeI05,SinghGB12}).
Bisubmodular functions can be minimized in polynomial time
given a value oracle over domains
of size $3$ \cite{FujishigeI05,qi88} but the complexity is
open on domains of larger size \cite{HuberK12}.

Fujishige and Murota \cite{FujishigeM00} introduced the notion
of $L^{\sharp}$-convex functions -- a class of functions that
can also be minimized in the oracle model \cite{Murota04}.  In
recent work, Kolmogorov \cite{Kolmogorov10} exhibited
efficient algorithms to minimize {\it strongly
  tree-submodular} functions on binary trees, which is a common
generalization of $L^{\sharp}$-convex functions and
bisubmodular functions.

\subsection{Technical Overview}

The technical heart of this paper is devoted to understanding the
evolution of probability distributions over $[q]^n$ under iterated
applications of operations. Fix a probability distribution $\mu$ over $[q]^n$.  For an operation
$p:[q]^k \to [q]$, the distribution $p(\mu)$ over $[q]^n$ is one that
is sampled by taking $k$ independent samples from $\mu$ and applying the operation
$p$ to them.  Fix a sequence $\{p_i\}_{i=1}^\infty$ of operations with
transitive symmetries.  Define the dynamical system,
$$ \mu_t = p_t(\mu_{t-1}) \mcom$$
with $\mu_0 = \mu$.  We study the properties of the distribution $\mu_t$ as $t
\to \infty$.  Roughly speaking, the key technical insight of this work is that the
correlations among the coordinates decay as $t \to \infty$.

For example, let us suppose $\mu_0$ is such that for each $i \in [n]$, the $i^{th}$
coordinate of $\mu_0$ is not perfectly correlated with the first
$(i-1)$-coordinates.  In this case, we will show that $\mu_t$
converges in statistical distance to a product distribution as $t \to
\infty$ (see \pref{thm:corr-decay}).  From an algorithmic standpoint,
this is very valuable because even if $\mu_0$ has no succinct
representation, the limiting distribution $\lim_{t \to \infty} \mu_t$
has a very succinct description.   Moreover, since the operations $p_i$ are applied to each
coordinate separately, the marginals of the limiting distribution
$\lim_{t \to \infty} \mu_t$ are determined entirely by the marginals of the initial distribution $\mu_0$.
Thus, since the limiting distribution is a product distribution, it is completely determined by
the marginals of the initial distribution $\mu_0$!

Consider an arbitrary probability distribution $\mu$ over
$[q]^n$.  Let $T_{1-\gamma} \circ \mu$ denote the probability distribution over
$[q]^n$ obtained by sampling from $\mu$ and perturbing each coordinate
with a tiny probability $\gamma$.  For small $\gamma$, the statistical
distance between $\mu$ and $T_{1-\gamma} \circ \mu$ is small, i.e.,
$ \norm{T_{1-\gamma} \circ \mu - \mu}_1 \leq \gamma n \mper$
However, if we initialize the dynamical system with $\mu_0 =
T_{1-\gamma} \mu$ then irrespective of the starting distribution
$\mu$, the limiting distribution is always a product distribution (see
\pref{cor:noisy-corr-decay}).
\iffull
A brief overview of the correlation decay argument is presented in
\pref{sec:linearity}.  The details of the argument are fairly
technical and draw upon various analytic tools such as
hypercontractivity, the Berry-Esseen theorem and Fourier analysis (see
\pref{sec:corr-decay}).  A key bottleneck in the analysis is that the
individual marginals change with each iteration thereby changing the
fourier spectrum of the operations involved.
\fi

Recall that the algebraic dichotomy conjecture for exact CSPs asserts
that a CSP $\Lambda$ admits a cyclic polymorphism if and only if the
CSP $\Lambda$ is in $\mathsf{P}$.  It is interesting to note that the
correlation decay phenomena applies to cyclic terms.  Roughly
speaking, the algebraic dichotomy
conjecture could be restated in terms of correlation decay in the above
described dynamical system.  This characterization closely resembles
the absorbing subalgebras characterization by Barto and Kozik
\cite{BartoK12} derived using entirely algebraic techniques.

Our approach to prove \prettyref{thm:main} is as follows.
For any function $f: [q]^n \to \R$, one can define a convex extension
$\hat{f}$.  Let $\ssimp_q$ denote the $q$-dimensional simplex.
\begin{definition} (Convex Extension)
  Given a function $f : [q]^n \to \R$, define its convex
  extension on $\hat{f}: \ssimp_q^n \to \R$ to be
  \begin{displaymath} \hat{f}(z) \defeq \min_{\substack{\text{ p.d.f. } \mu \text{
          over } [q]^n\\ \mu_i = z_i}} \E_{x \sim \mu} \left[f(x)\right]
  \end{displaymath}
  where the minimization is over all probability distributions
  $\mu$ over $[q]^n$ whose marginals $\mu_i$ coincide with
  $z_i$.
\end{definition}
The convex extension $\hat{f}(z)$ is the minimum expected value of $f$
under all probability distributions over $[q]^n$ whose marginals are
equal to $z$.  As the name suggests, $\hat{f}$ is
a convex function minimizing which is equivalent to minimizing $f$.
Since $\hat{f}$ is convex, one could appeal to techniques from
convex optimization such as the ellipsoid algorithm to minimize
$\hat{f}$.  However it is in general intractable to even
evaluate the value of $\hat{f}$ at a given point in $\ssimp_q^n$.
In the case of a submodular function $f$, its convex extension $\hat{f}$ can be evaluated at a given point via a greedy
algorithm, and $\hat{f}$ coincides with the function known as
the Lovasz-extension of $f$.

We exhibit a randomized algorithm to evaluate the value of the convex
extension $\hat{f}$ when $f$ admits a fractional polymorphism.
Given a point $z \in \ssimp_q^n$, the randomized algorithm computes
the minimum expected value of $f$ over all probability distributions
$\mu$ whose marginals are equal to $z$. Let $\mu$ be the
optimal probability distribution that achieves the minimum.
Here $\mu$ is an unknown probability distribution which might
not even have a compact representation.
Consider the probability distribution $T_{1-\gamma} \circ \mu$
obtained by resampling each coordinate independently with probability $\gamma$.  The
probability distribution $T_{1-\gamma} \circ \mu$ is statistically
close to $\mu$ and therefore has approximately the same expected value
of $f$.

Let $\mu'$ be the limiting distribution obtained by iteratively
applying the fractional polymorphism $\cP$ to the perturbed optimal
distribution $T_{1-\gamma} \circ \mu$.  Since $\cP$ is a fractional polymorphism, the expected value
of $f$ does not increase on applying $\cP$.
Therefore, the limiting probability distribution $\mu'$ has an
expected value not much larger than the optimal distribution $\mu$.
Moreover, since $\cP$ is measure preserving, the limiting probability
distribution has the same marginals as $\mu_0$!  In other words, the
limiting distribution $\mu'$ has marginals equal to $z$ and achieves
almost the same expected value as the
unknown optimal distribution $\mu$.
By virtue of correlation decay (\pref{cor:noisy-corr-decay}), the limiting
distribution $\mu'$ admits an efficient sampling algorithm that we use
to approximately estimate the value of $\hat{f}(z)$.

\section{Background} \label{sec:prelims}

We first introduce some basic notation. Let $[q]$ denote the
alphabet $[q] = \{1,\ldots, q\}$.  Furthermore, let $\ssimp_q$
denote the standard simplex in $\R^{q}$, i.e.,
$$ \ssimp_q = \{x\in \R^q | x_i \geq 0 \ \ \ \; \forall i,\  \sum_i x_i = 1
\} \mper$$
For a probability distribution $\mu$ on the finite set $[q]$ we
will write $\mu^k$ to denote the product distribution on $[q]^k$
given by drawing $k$ independent samples from $\mu$.

If $\mu$ is a joint probability distribution on $[q]^n$ we will write
$\mu_1,\mu_2,\dots\mu_n$ for the $n$ marginal distributions of $\mu$.
Further we will use $\mu^\times$ to denote the product distribution
with the same marginals as $\mu$. That is we define
\[
\mu^\times \defeq \mu_1 \times \mu_2 \times \dots \times \mu_n \mper
\]

An operation $p$ of arity $k$ is a map $p : [q]^k \to [q]$.  For a set
of $k$ assignments $x^{(1)},\ldots,x^{(k)} \in [q]^n$, we will use
$p(x^{(1)},\ldots,x^{(k)}) \in [q]^n$ to be the assignment obtained by
applying the operation $p$ on each coordinate of $x^{(1)},\ldots,x^{(k)}$
separately. More formally, let $x^{(i)}_j$ be the $j$th coordinate of $x_i$.
We define
\[
p(x^{(1)}\dots x^{(k)}) = \Big(p(x^{(1)}_{1}\dots x^{(k)}_{1}),p(x^{(1)}_{2}\dots x^{(k)}_{2}),
\dots,p(x^{(1)}_{n}\dots x^{(k)}_{n}) \Big) \mper
\]
\iffull
More generally, an operation can be thought of as a map $p : [q]^k \to
\ssimp_q$. In particular, we can think of $p$ being given by maps
$(p_1\dots p_q)$ where $p_i:[q]^k \to \R$ is the indicator
\[
p_i(x) =  \left\{
  \begin{array}{lr}
    1 & : p(x) = i\\
    0 & : p(x) \neq i\\
  \end{array}
\right.\\
\]
\fi

We next define a method for composing two $k$-ary operations $p_1$ and $p_2$. The idea is to think of each of $p_1$ and $p_2$ as nodes with $k$ incoming edges and one outgoing edge. Then we take $k$ copies of $p_2$ and connect those $k$ outputs to each of the $k$ inputs to $p_1$. Formally, we define:
\begin{definition}
  For two operations $p_1 : [q]^{k_1} \to [q]$ and $p_2 : [q]^{k_2} \to
  [q]$, define an operation $p_1 \otimes p_2 : [q]^{k_1 \times k_2} \to [q]$ as
  follows:
  \begin{displaymath} p_1 \otimes p_2 (\{ x_{ij}\}_{i \in [k_1],j\in [k_2]}) = p_1\left(
      p_2(x_{11},x_{12},\ldots,x_{1k_2}), \ldots,p_2(x_{k_1 1},x_{k_1
        2},\ldots,x_{k_1k_2}) \right)\end{displaymath}
\end{definition}

Next we state the definition for polymorphisms of an arbitrary cost function $f: [q]^n \to \R$. Intuitively, a polymorphism for $f$ is a probability distribution over operations that, on average, decrease the value of $f$.
\begin{definition}
  A $1$-approximate polymorphism $\cP$ for a function $f: [q]^n \to \R$, consists
  of a probability distribution $\cP$ over maps $\cO = \{ p :
  [q]^k \to [q]\}$ such that for any set of $k$ assignments $x^{(1)},\ldots,x^{(k)} \in [q]^n$,
  \begin{displaymath} \frac{1}{k} \sum_i f(x^{(i)}) \geq \E_{p \sim \cP}
    [f(p(x^{(1)},\ldots,x^{(k)}))]\end{displaymath}
\end{definition}

\begin{definition}
  For a $1$-approximate polymorphism $\cP$, $\cP^{\otimes r}$ denotes the
  $1$-approximate polymorphism consisting of a distribution over operations defined as,
  $p_1 \otimes \ldots \otimes p_r$
  with  $p_1, \ldots,p_r$ drawn i.i.d from $\cP$.
\end{definition}


\section{Correlation Decay} \label{sec:corrdecay}
In this section we state our main theorem regarding the decay of correlation between random variables under repeated applications of transitive symmetric operations.
\iffull
We begin by defining a quantitative measure of correlation and using it to bound the statistical distance to a product distribution.
\subsection{Correlation and Statistical Distance}
To gain intuition for our measure of correlation consider the example of two boolean random variables $X$ and $Y$ with joint distribution $\mu$. In this case we will measure correlation by choosing real-valued test functions $f,g:\{0,1\}\to \R$ and computing $\E[f(X)g(Y)]$. We would like to define the correlation as the supremum of $\E[f(X)g(Y)]$ over all pairs of appropriately normalized test functions. There are two components to the normalization. First, we require $\E[f(X)] = \E[g(Y)] = 0$ to rule out the possibility of adding a large constant to each test function to increase $\E[f(X)g(Y)]$. Second, we require $\Var[f(X)] = \Var[g(Y)] = 1$ to ensure (by Cauchy-Schwarz) that $\E[f(X)g(Y)] \leq 1$.

To see that this notion of correlation makes intuitive sense, suppose $X$ and $Y$
are independent. In this case correlation is zero because $\E[f(X)g(Y)] = \E[f(X)]\E[g(Y)] = 0$. Next suppose that $X = Y = 1$ with probability $\frac{1}{2}$ and $X = Y = 0$ with probability $\frac{1}{2}$. In this case we can set $f(1) = g(1) = 1$ and $f(0) = g(0) = -1$ to obtain $\E[f(X)g(Y)] = 1$. This matches up with the intuition that such an $X$ and $Y$ are perfectly correlated. We now give the general definition for our measure of correlation.

\begin{definition}
  Let $X,Y$ be discrete-valued random variables with joint distribution $\mu$.
  Let $\Omega_1 = ([q_1], \mu_1)$ and $\Omega_2 = ([q_2],\mu_2)$
  denote the probability spaces corresponding to $X,Y$ respectively.
  The correlation $\rho(X,Y)$ is given by
  \begin{displaymath} \rho(X,Y) \defeq
    \sup_{\substack{f:[q_1]
        \to \R, g:[q_2] \to \R\\ \Var[f(X)] = \Var[g(Y)] = 1 \\\E[f(X)] =
        \E[g(Y)] = 0 }} \E[f(X)g(Y)] \end{displaymath}
  We will interchangeably use the notation $\rho(\mu)$ or $\rho(\Omega_1,\Omega_2)$ to denote the correlation.
\end{definition}

Next we show that, as the correlation for a pair of random variables $X$ and $Y$ becomes small, the variables become nearly independent. In particular, we show that $\rho(X,Y)$ can be used to bound the statistical distance of $(X,Y)$ from the product distribution where $X$ and $Y$ are sampled independently.
\begin{lemma} \label{lem:corr-stat-distance}
  Let $X,Y$ be discrete-valued random variables with joint distribution $\mu_{XY}$ and respective marginal distributions $\mu_X$ and $\mu_Y$. If $X$ takes values in $[q_1]$ and $Y$ takes values in $[q_2]$, then
  \begin{displaymath} \norm{\mu_{XY} - \mu_X \times \mu_Y}_1 \leq \min(q_1,q_2) \rho(X,Y)\end{displaymath}
\end{lemma}
\begin{proof}
  Let $\{X_{a}\}_{a \in [q_1]}$ be indicator variables for
  the events $X = a$ with $a \in [q_1]$.  Similarly, define the
  indicator variables $\{Y_b\}_{b \in [q_2]}$.

  The statistical distance between $\mu_{XY}$ and $\mu_X \times
  \mu_Y$ is given by
  \begin{align*}
    \norm{\mu_{XY} - \mu_X \times \mu_Y}_1
    & = \sum_{a \in [q_1],b \in [q_2]} | \Pr[X = a, Y=b] -
      \Pr[X=a]\Pr[Y=b] | \\
    & = \sum_{a \in [q_1],b \in [q_2]} | \E[X_a Y_b]  -
      \E[X_a]\E[Y_b] |
  \end{align*}
  Set $\sigma_{ab} = \sign(\E[X_a Y_b]  -
  \E[X_a]\E[Y_b])$ and $Z_a = \sum_{b \in [q_2]} \sigma_{ab}
  Y_b$.
  \begin{align*}
    \norm{\mu_{XY} - \mu_X \times \mu_Y}_1
    & = \sum_{a \in [q_1]}  \E[X_a Z_a]  -   \E[X_a]\E[Z_a] \\
    & = \sum_{a \in [q_1]}  \Cov[X_a,Z_a] \\
    & \leq \sum_{a \in [q_1]}  \sqrt{\Var[X_a]} \sqrt{\Var[Z_a]}
      \rho(X,Y) \\
    & \leq \sum_{a \in [q_1]}  \sqrt{\Var[X_a]} \rho(X,Y) \qquad (\Var[Z_a] \leq 1)\\
    & \leq q_1 \rho(X,Y)
  \end{align*}
  The result follows by symmetry.

\end{proof}
%

\subsection{Correlation Decay}
\fi
To begin we give an explanation of why
one should expect correlation to decay under repeated applications of symmetric
operations. Consider the simple example of two boolean random variables $X$ and
$Y$ with joint distribution $\mu$. Suppose $X = Y$ with probability $\frac{1}
{2}+\gamma$ and that the marginal distributions of both $X$ and $Y$ are uniform
on $\{0,1\}$. Let $p:\{0,1\}^k \to \{0,1\}$ be the majority operation on $k$
bits. That is $p(x_1\dots x_k) = 1$ if and only if the majority of $x_1\dots
x_k$ are one.

Next suppose we draw $k$ samples $(X_i,Y_i)$ from $\mu$ and evaluate $p(X_1\dots
X_k)$ and $p(Y_1\dots Y_k)$. Since the marginal distributions of both $X$ and
$Y$ are uniform, the same is true for $p(X_1\dots X_k)$ and $p(Y_1\dots Y_k)$.
However, the probability that $p(X_1\dots X_k) = p(Y_1\dots Y_k)$ is strictly
less than $\frac{1}{2}+\gamma$. To see why first let $F:\{-1,1\} \to \{-1,1\}$
be the majority function where $1$ encodes boolean $0$ and $-1$ encodes boolean
$1$. Note that the probability that $F(X_1\dots X_k) = F(Y_1\dots Y_k)$ is given
by $\frac{1}{2} + \frac{1}{2}\E[F(X_1\dots X_k)F(Y_1\dots Y_k)]$.

Now if we write the Fourier expansion of $F$ the above expectation is
\[
\sum_{S,T} \hat{F}_S \hat{F}_T
\E\left[\prod_{i\in S} X_i \prod_{j\in T} Y_j\right]
= \sum_{S} \hat{F}_S^2 \prod_{i\in S} \E[X_i Y_i] = \sum_{S} \hat{F}_S^2 (2\gamma)^{|S|}
\]
Suppose first that all the non-zero Fourier coefficients $\hat{F}_S$ have
$|S| = 1$. In this case the probability that $F(X_1\dots X_k) = F(Y_1\dots Y_k)$
stays the same since $\frac{1}{2} + \frac{1}{2}(2\gamma) = \frac{1}
{2}+\gamma$. However, in the case of majority, it is well known that $\sum_{|S|
  = 1} \hat{F}_S^2 < 1 - c$ for a constant $c > 0$. Thus, the expectation is in fact
given by
\[
\E[F(X_1\dots X_k) F(Y_1\dots Y_k)] \leq (1-c)(2\gamma) + c(2\gamma)^2 <
2\gamma
\]
Thus the probability that $F(X_1\dots X_k) = F(Y_1\dots Y_k)$ is strictly less
than $\frac{1}{2} + \gamma$. Therefore, if we repeatedly apply the majority
operation, we should eventually have that $X$ and $Y$ become very close to
independent.

There are two major obstacles to generalizing the above observation to arbitrary
operations with transitive symmetry. The first is that for a general operation
$p$, we will not be able to explicitly compute the entire Fourier expansion.
Instead, we will have to use the fact that $p$ admits a transitive symmetry to
get a bound on the total Fourier mass on degree-one terms. The second obstacle
is that, unlike in our example, the marginal distributions of $X$ and $Y$ may
change after every application of $p$. This means that the correct Fourier basis
to use also changes after every application of $p$.

The fact that the marginal distributions change under $p$ causes difficulties
even for the simple example of the boolean OR operation on two bits.
Consider a highly biased distribution over ${0,1}$ given by $X = 1$ with probability $\eps$ and $X = 0$ with probability
$1-\eps$. Now consider the function $f(X) = \frac{1}{2}(X_1 + X_2)$. Note that
this function agrees with $OR$ except when $X_1 \neq X_2$. Thus, $f(X) = OR(X)$
with probability $1-2\eps(1-\eps) > 1 - 2\eps$. This means that as $\eps$
approaches zero, $OR$ approaches a function $f$ with
\[\sum_{|S| = 1} \hat{f}_S^2 = 1.
\]
Thus, there are distributions for which the correlation decay under the OR
operation approaches zero. This means that we cannot hope to prove a universal
bound on correlation decay for every marginal distribution, even in this very
simple case. It is useful to note that for the OR operation, the
probability that $X = 1$ increases under every application. Thus, as long as
the initial distribution has a non-negiligible probability that $X = 1$, we will
have that correlation does indeed decay in each step. Of course, this particular
observation applies only to the OR operation. However, our proof in the general
case does rely on the fact that, using only properties of the initial
distribution of $X$ we can get bounds on correlation decay in every step.

In summary,
we are able to achieve correlation decay for arbitrary transitive symmetric
operations. We now state our main theorem to this effect.
\begin{theorem} (Correlation Decay)\label{thm:corr-decay}
  Let $\mu$ be a distribution on $[q]^n$. Let $X_1,\ldots,
  X_n$ be the jointly distributed $[q]$-valued random variables
  drawn from $\mu$.  Let $\rho = \max_i \rho(X_1,\ldots, X_{i-1},
  X_i) < 1$ and $\lambda$ be the minimum probability of an atom in the
  marginal distributions $\{ \mu_i \}_{i \in [n]}$.
  For any $\eta > 0$, the following holds for $r \geq \Omega_q\left(\frac{\log\lambda}{\log\rho}
    \log^2\left(\frac{qn}{\eta}\right)\right)$: If $p_1,\ldots, p_r
  : [q]^k \to [q]$ is a sequence of operations each of which
  admit a transitive symmetry then,
  \[ \norm{ p_1 \otimes p_2 \otimes \ldots \otimes p_r (\mu) -
    p_1 \otimes p_2 \otimes \ldots \otimes p_r(\mu^{\times})}_1
  \leq \eta
  \]
\end{theorem}
We defer the proof of the theorem to \pref{sec:correlationdecay}.  Note that the theorem only applies when $\rho < 1$ i.e. when there are no perfect correlations between the $X_i$. To ensure that a distribution has no perfect correlations we can introduce a small amount of noise.
\iffull
\paragraph{Noise} For a probability distribution $\mu$ on $[q]^n$ let
$T_{1-\gamma}\circ\mu$ denote the probability distribution over
$[q]^n$ defined as:
\begin{itemize} 
\item Sample $X \in [q]^n$ from the distribution $\mu$.
\item For each $i \in [n]$, set
  \begin{displaymath}\tilde{X}_i = \begin{cases} X_i & \text{ with
        probability} 1- \gamma \\
      \text{sample from } \mu_i & \text{ with
        probability } \gamma
    \end{cases} \end{displaymath}
\end{itemize}
\fi
We prove the following corollary that gives correlation decay for distributions with a small amount of added noise.
\begin{corollary} \label{cor:noisy-corr-decay}
  Let $\mu$ be any probability distribution over $[q]^n$ and let
  $\lambda$ denote the minimum probability of an atom in the
  marginals $\{ \mu_i \}_{i \in [n]}$.  For any $\gamma,\eta > 0$,
  given a sequence $p_1,\ldots, p_r :[q]^k \to [q]$ of
  operations with transitive symmetry with $r >
  \Omega_q\left(\frac{\log\lambda}{\log(1-\gamma)}
    \log^2\left(\frac{qn}{\eta}\right)\right)$
  $$ \norm{ p_1 \otimes p_2 \otimes \ldots \otimes p_r
    (T_{1-\gamma} \circ \mu) - p_1 \otimes p_2 \otimes
    \ldots \otimes p_r (\mu^{\times}) }\leq \eta $$
\end{corollary}
\iffull
\begin{proof}
  Let $f:[q]\to\R$ and $g:[q]^{i-1} \to \R$ be functions with
  $\E[f(\tilde{X}_i)] = \E[g(\tilde{X_1}\dots\tilde{X}_{i-1})] = 0$ and
  $\E[f(\tilde{X}_i)^2] = \E[g(\tilde{X_1}\dots\tilde{X}_{i-1})^2] = 1$ such
  that
  \[
  \E[f(\tilde{X}_i)g(\tilde{X_1}\dots\tilde{X}_{i-1})] =
  \rho(\tilde{X}_i,(\tilde{X_1}\dots\tilde{X}_{i-1}))
  \]
  That is, $f$ and $g$ achieve the maximum possible correlation between
  $\tilde{X}_i$ and $(\tilde{X_1}\dots\tilde{X}_{i-1})$. Let $Y_i$ be an
  independent random sample from the marginal $\mu_i$. Now we expand the
  above expectation by conditioning on whether or not $\tilde{X}_i$ was
  obtained by re-sampling from the marginal $\mu_i$.
  \[
  \E[f(\tilde{X}_i)g(\tilde{X_1}\dots\tilde{X}_{i-1})] =
  \E[f(X_i)g(\tilde{X_1}\dots\tilde{X}_{i-1})](1 - \gamma) +
  \E[f(Y_i)g(\tilde{X_1}\dots\tilde{X}_{i-1})]\gamma
  \]
  By Cauchy-Schwarz inequality, the first term above is bounded by
  \[
  \E[f(X_i)g(\tilde{X_1}\dots\tilde{X}_{i-1})](1 - \gamma) \leq
  \sqrt{\E[f(X_i)^2]\E[g(\tilde{X_1}\dots\tilde{X}_{i-1})^2]}(1 - \gamma)
  = 1 - \gamma
  \]
  where we have used the fact that $\E[f(X_i)^2] = \E[f(\tilde{X}_i)^2] = 1$.
  Since $Y_i$ is independent of $\tilde{X_1}\dots\tilde{X}_{i-1}$, the
  second term is
  \[
  \E[f(Y_i)g(\tilde{X_1}\dots\tilde{X}_{i-1})]\gamma =
  \E[f(Y_i)]\E[g(\tilde{X_1}\dots\tilde{X}_{i-1})]\gamma = 0
  \]
  Thus, we get $\rho(\tilde{X}_i,(\tilde{X_1}\dots\tilde{X}_{i-1})) \leq 1 -
  \gamma$ for all $i \in [n]$. The result then follows by applying Theorem
  \ref{thm:corr-decay} to $T_{1-\gamma} \circ \mu$.
\end{proof}
\fi
 \label{sec:linearity}
\section{Optimization in the Value Oracle Model}

	In this section, we will describe an algorithm to minimize a
	function $f: [q]^n \to \R$ that admits a $1$-approximate
	polymorphism given access to a value oracle for $f$.  We begin
	by setting up some notation.  

	Recall that for a finite set $A$, $\ssimp_A$ denotes the
	set of all probability distributions over $A$.
	For notational convenience, we will use the following
	shorthand for the expectation of the function $f$ over a
	distribution $\mu$. 
\begin{definition} (Distributional Extension)
	Given a function $f: [q]^n \to \R$, define its distributional
	extension $F: \ssimp_{[q]^n} \to \R$ to be
	\begin{displaymath} F(\mu) \defeq \E_{x \sim \mu} \left[f(x)\right] \end{displaymath}
\end{definition}
Given an operation $p: [q]^k \to [q]$ and a probability distribution
$\mu \in \ssimp_[q]^n$, define $p(\mu)$ to be the
distribution over $[q]^n$ that is sampled as follows:
\iffull
\begin{itemize} \itemsep0pt
	\item Sample $x_1,x_2, \ldots,x_k \in [q]^n$ independently from the
		distribution $\mu$.
	\item Output $p(x_1,x_2,\ldots,x_k)$.
\end{itemize}
\fi
Notice that for each coordinate $i \in [n]$, the $i^{th}$ marginal distribution
$(p(\mu))_i$ is given by $p(\mu_i)$.  More generally, if a $\cP$
denotes a probability distribution over operations then $\cP^{\otimes
r}(\mu)$ is
a distribution over $[q]^n$ that is sampled as follows:
\iffull
\begin{itemize} \itemsep0pt
	\item Sample operations $p_1,\ldots, p_r : [q]^k \to [q]$
		independently from the
		distribution $\cP$
	\item Output a sample from $p_1 \otimes p_2 \otimes \ldots p_r (\mu)$.
\end{itemize}
\fi

Suppose $\cP$ is a fractional polymorphism for the function $f$.  By
definition of fractional polymorphisms, the average value of $f$ does
not increase on applying the operations sampled from $\cP$.  More
precisely, we can prove the following.
\begin{lemma} \label{lem:frac-poly-dist}
For every distribution $\mu$ on $[q]^n$, and a fractional polymorphism $\cP$ of a
function $f : [q]^n \to \R$ and $r \in \N$,
$ F(\cP^{\otimes r}(\mu)) \leq F(\mu) $
where $F$ is the distributional extension of $f$
\end{lemma}
\iffull
\begin{proof}
	We will prove this result by induction on $r$.  First, we prove
	the result for $r = 1$.
  	\begin{align*}
		F(\cP(\mu)) & =  \E_{x \sim \cP(\mu)} \left[ f(x) \right] \\
	& =  \E_{p \sim \cP} \E_{x_1,\ldots, x_k \sim \mu} \left[
		f(p(x_1,\ldots, x_k)) \right] \\
	& =   \E_{x_1,\ldots, x_k \sim \mu} \left[\E_{p \sim \cP} \left[
		f(p(x_1,\ldots, x_k)) \right]\right] \\
		&\leq \E_{x_1,\ldots, x_k \sim \mu} \left[
			\frac{1}{k} \sum_i f(x_i) \right]  = F(\mu)
	\end{align*}
	The last inequality in the above calculation uses the fact
	that $\cP$ is a fractional polymorphism for $f$.	
	Suppose the assertion is true for $r$, now we will prove the
	result for $r+1$.  Observe that the distribution $\cP^{\otimes
	r+1} (\mu)$ can be written as,
	$$ \cP^{\otimes r+1}(\mu)  = \E_{p_2,\ldots, p_{r+1} \sim \cP}
	\left [
		\cP( p_2 \otimes \ldots \otimes p_r(\mu)) \right]
		$$
	where $p_2,\ldots, p_r$ are drawn independently from the
	distribution $\cP$.  Hence we can write,
	\begin{align*}
	F( \cP^{\otimes r+1}(\mu)) &= \E_{p_2,\ldots, p_{r+1}\sim
\cP} \left[
		F(\cP(p_2 \otimes \ldots \otimes p_{r+1} (\mu)))
		\right] \\
&	\leq	 \E_{p_2,\ldots, p_{r+1}\sim
\cP} \left[
		F(p_2 \otimes \ldots \otimes p_{r+1} (\mu))
		\right] \qquad \text{ using base case}\\
	&= F( \cP^{\otimes r}(\mu))  \leq F(\mu) \mcom
	\end{align*}
	where the last inequality used the induction hypothesis for
	$r$.
\end{proof}
\fi

	%
	%
	%
	%
	%
	%

Recall that $\mu^{\times}$ is the product distribution with the same
marginals as $\mu$.  We will show the following using correlation
decay.

\begin{lemma} \label{lem:main}
	Let $\cP$ be a fractional polymorphism with a transitive symmetry for a function $f :
	[q]^n \to [q]$.  Let $\mu$ be a probability distribution over 
	$[q]^n$ and let $\lambda$ denote the minimum probability of an atom in the
	marginals $\{ \mu_i \}_{i \in [n]}$. For $\gamma = \frac{\delta}{2n}$ and $r 
	= \Omega_q\left(\frac{\log \lambda}{\log(1-\gamma)} \log^2\left(\frac{2qn}
	{\delta}\right)\right)$ 
	\begin{displaymath}F(\cP^{\otimes r}(\mu^{\times}))  \leq
		F(\mu) +  \delta \norm{f}_{\infty} \end{displaymath}
\end{lemma}
\iffull
\begin{proof}
	Consider the distribution $T_{1-\gamma} \circ \mu$.  By
	definition of $T_{1-\gamma} \circ \mu$, all the correlations
	within $T_{1-\gamma} \circ \mu$ are at most $1-\gamma$.
	Roughly speaking, with repeated applications of the operations from $\cP$
	all the correlations will vanish.  More precisely, by
	\prettyref{cor:noisy-corr-decay}, for any sequence of
	operations $p_1,\ldots, p_r$ with transitive symmetry we have
	\begin{equation}\label{eq:tensorpoly1}
	\norm{ p_1 \otimes p_2 \otimes \ldots \otimes p_r
		(T_{1-\gamma} \circ \mu) - p_1 \otimes p_2 \otimes
		\ldots \otimes p_r (\mu^{\times}) }\leq
		\frac{\delta}{2} 
	\end{equation}
	Recall that $\cP^{\otimes r}(\mu')$ for a distribution $\mu'$ is given by,
	$ \cP^{\otimes r}(\mu') = \E_{p_1,\ldots, p_r} \left[p_1 \otimes p_2
	\ldots \otimes p_r (\mu') \right]$.  Averaging
	\eqref{eq:tensorpoly1} over all choices of $p_1,\ldots p_r$
	from $\cP$ we get
	\begin{equation} \label{eq:tensorpoly2}
	 \norm{ \cP^{\otimes r}
	(T_{1-\gamma} \circ \mu) - \cP^{\otimes r} (\mu^{\times})
}\leq \frac{\delta}{2}
	\end{equation} 
Now we are ready to finish the proof of the lemma.	
	\begin{align*}
F(\cP^{\otimes r}(\mu^{\times})) %
	&\leq
	F(\cP^{\otimes r}(T_{1-\gamma} \circ \mu )) + \frac{\delta}{2}
	\norm{f}_{\infty} \qquad \text{ using 
		(\eqref{eq:tensorpoly2})}\\
		& \leq F(T_{1-\gamma} \circ \mu) +  \frac{\delta}{2}
		\norm{f}_{\infty}  \qquad \text{ using
			\prettyref{lem:frac-poly-dist}} \\
		& \leq F(\mu) + \delta \norm{f}_\infty \qquad \text{
		using } \norm{\mu - T_{1-\gamma} \circ \mu}_1 \leq
		\gamma n
\end{align*}
\end{proof}
\fi
Suppose we are looking to minimize the function $f$ or equivalently
the distributional extension $F$.  In general, the minima for $F$
could be an arbitrary distribution with no succinct (polynomial-sized)
representation.  The preceding lemma shows that $\cP^{\otimes
r}(\mu^{\times})$ has roughly the same value of $F$.  But
$\cP^{\otimes r}(\mu^{\times})$ not only has a succinct representation,
but is efficiently sampleable given the marginals of $\mu$!  In the
rest of this section, we will use this insight towards minimizing $f$
given a value oracle.
\iffull
\subsection{Convex Extension}  
\fi
\begin{definition} (Convex extension)
	For a function $f: [q]^n \to \R$, the convex extension
	$\hat{f}: (\ssimp_q)^n \to \R$ is defined as,
	$$ \hat{f}(w) \defeq \min_{ \substack{\mu \in \ssimp_{[q]^n}\\ \mu_i
= w_i \forall i \in [n]}}  E_{x \in \mu} [f(x)] \mcom $$
	where the minimization is over all probability distributions
	$\mu$ over $[q]^n$ whose marginals are given by $w \in
	(\ssimp_q)^n$.
\end{definition}
As the name suggests, $\hat{f}$ is a convex function whose minimum is
equal to the minimum of $f$.  In general, the convex extension of a function $f$ cannot be computed
efficiently since the optimal distribution $\mu$ might not even have a
succinct representation. In our case, however, we can prove:

\begin{theorem}\label{thm:convex-estimate}
	Suppose a function $f : [q]^n \to [q]$ admits a fractional
	polymorphism $\cP$ such that: (1) Each operation $p : [q]^k \to [q]$ in the support of $\cP$ has a transitive symmetry, (2) the polymorphism $\cP$ is measure preserving.
	Then there is an algorithm that given $\epsilon >0$ and $w \in
	\ssimp_q^n$, runs in time $\poly(n,\frac{1}{\epsilon})$ and
	computes $\hat{f}(w) \pm \epsilon \norm{f}_\infty$.
\end{theorem}
\iffull
\begin{proof}
	Given $w \in \ssimp_q^n$, we first perturb every coordinate slightly to 
	ensure that the minimum probability in each marginal $w_i$ is bounded away 
	from $0$. In particular, we define $w'$ by setting $b = \frac{\eps}{10qn}$ 
	and
	\[
	w'_i(a) = \frac{w_i(a) + b}{1 + qb}
	\]
	for all $a\in [q]$ and $i\in [n]$. Clearly $w' \in \ssimp_q^n$ and $w'_i \geq 
	\frac{b}{1 + qb}$. Now let 
	$$ \mu = \underset{\substack{\mu \in
		\ssimp_{[q]^n}\\ \mu_i
= w'_i \forall i \in [n]}}{\operatorname{argmin}}   E_{x \in \mu} [f(x)] \mcom $$
	denote the
	optimal distribution over $[q]^n$ that achieves the minimum in
	$\hat{f}(w')$.  Fix $\delta = \frac{\eps}{10}$ and note that
	\[
	\frac{\log b}{\log (1 - \frac{\delta}{2n})} = O_q\left(n\eps^{-1}\log(n\eps^{-1})\right)
	\]	
	Now we claim that for $r = \Omega\left(n\eps^{-1}\log^3(n\eps^{-1})\right)$ 
	we have
	$$ F(\cP^{\otimes r}(\mu^{\times})) - \delta \norm{f}_\infty
	\leq \hat{f}(w') = F(\mu) \leq
	F(\cP^{\otimes r}(\mu^{\times}) \mper$$
	The left-hand inequality follows from \prettyref{lem:main}.
	The right hand side inequality follows because
	$\cP^{\otimes r} (\mu^{\times})$ has its marginals equal to
	$w'$ since $\cP$ is {\it measure preserving} and $\mu$ is
	the optimal distribution with marginals $w'$.  Moreover, observe
	that the distribution $\cP^{\otimes r}(w')$ can be sampled
	efficiently as described below.
\vspace{10pt}
\begin{mybox}
	Input: $w' \in \ssimp_q^n$\\
	Output:  Sample from $\cP^{\otimes r}(w')$
\begin{itemize} \itemsep=0ex
		\item	Fix $z_0 = w'$
		\item For $i = 1$ to $r$ do,
		\begin{itemize}\itemsep=0ex
			\item Sample $p \in \cP$ and set $z_i =
				p(z_{i-1})$.
			\end{itemize}
		\item Sample $x \in [q]^n$ by sampling each coordinate
			independently from $z_r$.
	\end{itemize}

\end{mybox}
\vspace{15pt}
	
	In order to estimate $\hat{f}(w')$, it is sufficient to sample
	$\cP^{\otimes r}(w)$ independently $O\left(\frac{n}{\eps^2}\right)$ times to 
	compute $F(\cP^{\otimes r}(\mu^{\times}))$ to accuracy within $\frac{\eps}
	{10}\norm{f}_\infty$ with high probability. Thus, we can estimate $\hat{f}
	(w')$ to accuracy $\left(\delta + \frac{\eps}{10}\right)\norm{f}_\infty = 
	\frac{\eps}{5}\norm{f}_\infty$. 
	
	Next let $\pi$ be the distribution that achieves the optimum for marginals 
	$w$ i.e. $\hat{f}(w) = F(\pi)$. By changing each marginal of $\pi$ by at most 
	$b$ we obtain a distribution $\pi'$ with marginals given by $w'$. By 
	optimality of $\mu$ we know
	\[
	\hat{f}(w') = F(\mu) \leq F(\pi') = F(\pi) \pm bqn\norm{f}_\infty = 
	\hat{f}(w) \pm \frac{\eps}{10}\norm{f}_\infty
	\]
	By a symmetric argument we get that $\hat{f}(w) \leq \hat{f}(w') \pm 
	\frac{\eps}{10}\norm{f}_\infty$. Thus, we conclude that 
	$|\hat{f}(w) - \hat{f}(w')| \leq \frac{\eps}{10}\norm{f}_\infty$. 
	Therefore we can estimate $\hat{f}(w)$ to accuracy 
	$(\frac{\eps}{10} + \frac{\eps}{5})\norm{f}_\infty \leq 
	\eps\norm{f}_\infty$. 
	In summary, 	this yields an algorithm running in time 
	$O_q\left(n^2\eps^{-3}\log^3(n\eps^{-1}) \right)$ to estimate $\hat{f}(w)$
	within an error of $\eps\norm{f}_\infty$ with high probability.
\end{proof}
\fi

\paragraph{Gradient-Free Minimization}

In the previous section, we have demonstrated that the convex
extension $\hat{f}: \ssimp_q^n \to \R$ can be computed efficiently.
In order to complete the proof of our main theorem (\pref{thm:main}),
we will exhibit an algorithm to minimize the convex function
$\hat{f}$.
The domain of the convex function $\hat{f}$, namely $\ssimp_q^n$ is
particularly simple.  However, the convex function $\hat{f}$ is not
necessarily smooth (differentiable).  More importantly, we only have
an evaluation oracle for the function $\hat{f}$ with no access to its
gradients or subgradients.
Gradient-free optimization has been extensively studied (see \cite{Nesterov11,Spall97} and
references therein) and there are numerous approaches that have been
proposed in literature.
At the outset, the idea is to estimate the gradient via one or more
function evaluations in the neighborhood.
In this work, we will appeal to the randomized algorithms by gaussian
perturbations proposed by Nesterov \cite{Nesterov11}.
The following result is a restatement of Theorem $5$ in the work of
Nesterov \cite{Nesterov11}. 

\newcommand{\diam}{\Delta}
\begin{theorem} \label{thm:nesterov}
	Let $\hat{f}: \R^N \to \R$ be a non-smooth convex function given by an
evaluation oracle $\cO$.  Let us suppose $f$ is
$L$-Lipshitz,i.e.,
$$ |\hat{f}(x) - \hat{f}(y)| \leq L \norm{x-y} \qquad \forall x,y \in \R^N
$$
Let $Q \sse \R^N$ be a closed convex set
given by a projection oracle $\pi_Q : \R^N \to Q$.  Let 
$\diam(Q) \defeq \sup_{x,y \in Q} \norm{x-y}$.

There is a randomized algorithm that given $\epsilon > 0$, with
constant probability finds $y \in
Q$ such that,
$$ |\hat{f}(y) - \min_{x \in Q} \hat{f}(x)| \leq \e\mcom$$
using $$\frac{ (N+4)L^2 \diam(Q)^2}{\epsilon^2}$$ calls to the oracle
$\cO$ and $\pi_Q$. 
\end{theorem}
In our setup, the domain $Q = \ssimp_q^n \in \R^{nq}$ is a very
simple closed convex set with diameter $\Delta(Q) \leq \sqrt{2n}$.
The convex extension $\hat{f}$ is $L$-Lipschitz for $L =
\norm{f}_\infty$.   
Finally, we turn to the proof of our main theorem, which follows more or less 
immediately from the results in this section.
\begin{proof}[Proof of \prettyref{thm:main}]
By \prettyref{thm:convex-estimate} we can estimate the convex 
function $\hat{f}(w)$ with sufficient accuracy to apply 
\prettyref{thm:nesterov} to approximately minimize $\hat{f}(w)$ in time 
$\poly(1/\epsilon,n,\norm{f}_{\infty})$. Since $\hat{f}(w)$ is 
the convex extension of $f$, the output of the algorithm approximately minimizing 
$\hat{f}(w)$ will also approximately minimize $f$.
\end{proof}

\section{Analytic Background}
In this section we introduce the analytic background necessary to prove Theorem \ref{thm:corr-decay}.
\subsection{Analysis on Finite Measure Spaces}
We begin by introducing the space of functions that will be fundamental to our analysis.
\begin{definition}
  Let $\mu$ be a probability distribution on $[q]$. Define $L^2([q],\mu)$ to be the inner product space of all functions $f:[q] \to \R$ with inner product given by
  \[
  \langle f,g \rangle = \E_{x\sim\mu}[f(x)g(x)] \mper
  \]
\end{definition}

For a function $f\in L^2([q],\mu)$ the $L^p$-norm of $f$ is defined as
\[
\norm{f}_p = \E_{x\sim\mu}[f(x)^p]^{\frac{1}{p}} \mper
\]
We will also often use the equivalent notation $L^2(\Omega)$ where $\Omega = ([q],\mu)$ is the corresponding probability space. We next introduce a useful basis for representing functions $f:[q]^k \to \R$.

\paragraph{Multilinear Representation}

Let $\mu$ be a distribution on $[q]$.  Pick an orthonormal basis
$\phi = \{\phi_0 \equiv 1, \ldots, \phi_{q-1}\}$ for the set of functions from
$[q]$ to $\R$ denoted by $L^2([q],\mu)$.

The tensor product of the orthonormal basis $\phi$ gives an orthonormal basis for
functions in $L^2([q]^k, \mu^k)$.  Specifically, for a multi-index $\alpha \in
\{0,1,\ldots,q-1\}^k$, let $\phi_{\alpha} : [q]^k \to \R$ denote the function,
$\phi_{\alpha}(x) = \prod_{i = 1}^k \phi_{\alpha_i}(x_i)$.  It is easy
to see that $\phi_{\alpha}$ forms an orthonormal basis for space of
functions $L^2([q]^k, \mu^k)$.

We can write any function $p : [q]^k \to \R$ in the $\phi_\alpha$ basis as
\begin{displaymath}p(x) = \sum_{\alpha \in \{0,1,\ldots,q-1\}^k} \hat{p}_{\alpha} \phi_{\alpha}(x)
\end{displaymath}
where we often refer to the $\hat{p}_\alpha$ as the Fourier coefficients of $p$. For a multi-index $\alpha \in \{0,1,\ldots, q-1\}^k$,
let $|\alpha| = |\{ i | \alpha_i \neq 0\}|$.  The degree of a term
$\phi_{\alpha}(x)$ is given by $|\alpha|$.

Let $p : [q]^k \to \ssimp_q$ be a $k$-ary operation represented by
$q$-real valued functions $p = (p_1,\ldots, p_{q})$.  Associated with every $k$-ary operation $p$ is a subspace of $L^2(\mu^k)$
given by
\begin{displaymath} \Sp(p) = \{ \text{ span of } p_1,\ldots, p_q \}\end{displaymath}

\paragraph{Noise Operator}
For a function $p : [q]^k \to \R$ and $\rho \in [0,1]$, define $T_{\rho} p$ as
\begin{displaymath} T_{\rho} p (x) = \E_{y \sim_{\rho} x}[p(y)]\end{displaymath}
where $y \sim_{\rho} x$ denotes the following distribution,
$$ y_i = \begin{cases} x_i & \text{ with probability } \rho \\ \text{
    indpendent sample from } \mu & \text{ with probability } 1-\rho
\end{cases}$$
The multilinear polynomial associated with $T_{\rho} p$ is given by
\begin{displaymath} T_{\rho} p =  \sum_{\alpha} \hat{p}_{\alpha}
  \rho^{|\alpha|} \phi_{\alpha}(x)
\end{displaymath}

\subsection{The Conditional Expectation Operator}
In this section we introduce the conditional expectation operator associated with
a joint probability distribution. The singular values of this operator encode
information about correlation, and thus the singular vectors provide a useful
basis in which to analyze correlation decay.

Let $\mu$ be a joint distribution on $[q_1] \times [q_2]$ with marginals $\mu_1$
and $\mu_2$. Further let $\Omega_1 = ([q_1],\mu_1)$ and $\Omega_2 = ([q_2],
\mu_2)$ denote the probability spaces corresponding to $\mu_1$ and $\mu_2$.
\begin{definition}
  The conditional expectation operator $T_\mu:L^2(\Omega_2)\to
  L^2(\Omega_1)$ is given by
  \[
  (T_\mu f)(x) = \E_{(X,Y)\sim\mu}[f(Y) | X = x]
  \]
\end{definition}
It follows from the definition that
\[
\E_{(X,Y)\sim\mu}[f(X)g(Y)] = \E_{X\sim\mu_1}[f(X)(T_\mu g) (X)]
\]
The adjoint operator $T_\mu^*:L^2(\Omega_1) \to L^2(\Omega_2)$ is given by
\[
(T^*_\mu g)(y) = \E_{(X,Y)\sim\mu}[g(X) | Y = y]
\]
It is possible to choose a singular value decomposition $\{\phi_i,\psi_j\}$ of
the operator $T_{\mu}$ such that $\phi_0 = \psi_0 \equiv 1$. Let $\sigma_i$
denote the singular value corresponding to the pair $\phi_i,\psi_i$. That is
$T_\mu \psi_i = \sigma_i \phi_i$. It is easy to check the following facts
regarding this singular value decomposition.

\begin{fact}
  Let $\{\phi_i,\psi_j\}$ be the above singular value decomposition of $T_\mu$.
  \begin{itemize}
  \item The functions $\phi_i$ and $\psi_i$ form orthonormal bases for
    $L^2(\Omega_1)$ and $L^2(\Omega_2)$ respectively.
  \item $\sigma_0 = 1$ and $\sigma_1 = \rho(\mu)$.
  \item The expectation over $\mu$ of the product of any pair $\phi_i$ and $\psi_i$
    is given by
    \[
    \E_{(x,y)\sim\mu}[\phi_i(x)\psi_j(y)] =  \left\{
      \begin{array}{lr}
        \sigma_i & : i = j\\
        0 & : i \neq j
      \end{array}
    \right.\\
    \]
  \end{itemize}
\end{fact}
Thus, for functions $F\in L^2([q_1]^k,\mu_1^k)$ and $G\in L^2([q_2]^k,\mu_2^k)$
we can write their multilinear expansions with respect to tensor powers of the
$\phi$ and $\psi$ bases respectively. In particular, for a multi-index $\alpha
\in \{0,\dots, q_1-1\}^k$ and $x\in [q_1]^k$ we have the basis functions
\[
\phi_\alpha (x) = \prod_i \phi_{\alpha_i}(x_i)
\]
with $\psi_\beta(y)$ defined similarly. We then have the following fact
\begin{fact}
  The vectors $\{\phi_\alpha\}_{\alpha \in [q_1]^k}$, $\{\psi_\beta\}_{\beta \in [q_2]^k}$ are a singular value decomposition of $T_\mu^{\otimes k}$ with corresponding singular values $\sigma_\alpha = \prod_i \sigma_{\alpha_i}$. Furthermore $T_{\mu^k} = T_\mu^{\otimes k}$.
\end{fact}
Using the above facts about the singular value basis it follows that
\[
\E_{(x,y)\sim\mu^k}[\phi_\alpha(x)\psi_\beta(y)] =  \left\{
  \begin{array}{lr}
    \sigma_\alpha & : \alpha = \beta\\
    0 & : \alpha \neq \beta
  \end{array}
\right.\\
\]
We can then write the multi-linear expansion of $F$ with respect to $\phi$ as
\[
F(x) = \sum_\alpha \hat{F}_\alpha \phi_\alpha(x)
\]
As an immediate consequence of the above discussion we have
\begin{fact}
  Let $F\in L^2([q_1]^k,\mu_1^k)$ and $G\in L^2([q_2]^k,\mu_2^k)$. Then
  \[
  \E_{(x,y)\sim\mu^k}[F(x)G(y)] = \sum_\alpha \hat{F}_\alpha \hat{G}_\alpha \sigma_\alpha
  \]
\end{fact}

\section{Correlation Decay for Symmetric Polymorphisms}
\label{sec:correlationdecay}

In this section we prove \pref{thm:corr-decay}. The proof has two major components. First, we show that for an operation where the Fourier weight on certain degree-one coefficients is bounded away from one, correlation decreases. Second, we show that for operations admitting a transitive symmetry, this Fourier weight does indeed stay bounded away from one after each application of an operation.

Throughout the section let $\mu$ be a joint distribution on $[q]\times [q]^m$ for some $m$. Let $\mu_1$ and $\mu_2$ be the respective marginal distributions of $\mu$. Further, let $\Omega_1 = ([q], \mu_1)$ and $\Omega_2 = ([q]^m,\mu_2)$ denote the probability spaces corresponding to $\mu_1$ and $\mu_2$. It will be useful to think of $q$ as being a constant and $m$ as possibly being very large.

Further, for an operation $p:[q]^k \to [q]$ we will write $p(\mu)$ for the probability distribution obtained by sampling $(X,Y) \sim \mu^k$ and outputting the pair $(p(X_1,\dots,X_k),p(Y_1,\dots,Y_k)) \in [q]\times [q]^m$. Recall that we have defined $p(Y_1,\dots,Y_k)$ for $Y_i \in [q]^m$ to be the result of applying $p$ on each coordinate of $Y_1,\dots,Y_k$ separately.

\subsection{Linearity and Correlation Decay}
Given a singular value basis $\{\phi_\alpha\}$ for the space $L^2([q]^k,\mu_1^k)$ we call the set of functions $\phi_\alpha$ such that $|\alpha| = 1$ (i.e. exactly one coordinate of $\alpha$ is nonzero) the degree-one part of the basis. Note that there are $qk$ such basis functions: $q$ for each of the $k$
coordinates. We will be interested in keeping track of those degree-one basis
functions with singular values close to the correlation $\rho(\mu)$.
\begin{definition}
  Let $\{\phi_\alpha\}$ be a singular value basis as above. The linear
  multi-indices $\alpha$ are given by the set
  \[
  L_\mu = \{\alpha \mid |\alpha| = 1 \mbox{ and } \sigma_\alpha \geq \rho(\mu)^2\}
  \]
  The linear part of a singular value basis is those $\phi_\alpha$ with $\alpha \in L_\mu$.
\end{definition}
Next we define a quantitative measure of the linearity of a function.
\begin{definition}
  Let $F\in L^2([q]^k,\mu_1^k)$. The linearity of $F$ is defined as
  \[
  \Lin_\mu(F) = \sum_{\alpha\in L_\mu} \hat{F}_\alpha^2
  \]
  Further, for an operation $p:[q]^k \to [q]$ the linearity of $p$ is given by
  \[
  \Lin_\mu(p) = \sup_{\substack{F\in \Sp(p) \\ \norm{F}_2 = 1}} \Lin_\mu(F)
  \]
\end{definition}
With these definitions in hand we show that the correlation of $\mu$ decays under any operation $p$ that has linearity bounded away from one.
\begin{lemma} (Correlation Decay) \label{lem:cordecaysym}
  Let $p:[q]^k \to [q]$ be an operation. Then
  \[
  \rho(p(\mu)) \leq \rho(\mu)\left( 1- \frac{1}{2}(1-\Lin_\mu(p))(1-\rho(\mu)^2) \right)
  \]
\end{lemma}
\begin{proof}
  Let $f : [q] \to \R$ and $g : [q]^m \to \R$ be such that $\E_{x\sim\mu_1^k}[f(p(x)] =
  0$, $\E_{y\sim\mu_2^k}[g(p(y))] = 0$, $\norm{f(p)}_2 = \norm{g(p)}_2 = 1$ and
  \[
  \rho(p(\mu)) = \E_{(x,y)\sim\mu^k}[f(p(x))g(p(y))] \mper
  \]
  That is, $f$ and $g$ are the functions achieving the maximum correlation for
  $p(\mu)$. Now define $F: [q]^k \to \R$ as $F(x) = f(p(x))$. Similarly,
  define $G: [q]^{mk} \to \R$ as $G(y) = g(p(y))$.

  Writing the multilinear expansion of $F$ with respect to the basis $\phi$, we
  have $F = \sum_{\alpha}
  \hat{F}_\alpha \phi_{\alpha}$ where
  \[
  \hat{F}_{0} = \E[F] = 0\mcom \qquad \sum_{\alpha} \hat{F}_{\alpha}^2
  = \norm{F}_2 = 1 \mper
  \]
  Furthermore, since $F\in\Sp(p)$ we have
  \[
  \sum_{\alpha\in L_\mu} \hat{F}_{\alpha}^2 = \Lin_\mu(F) \leq \Lin_\mu(p) \mper
  \]
  Similarly, one can write the multilinear expansion of $G$.  Recall that,
  \[
  \E_{(x,y)\sim\mu^k}[f(p(x))g(p(y))] = \E_{(x,y)\sim\mu^k}[F(x) G(y)] =
  \sum_{\alpha} \hat{F}_\alpha \hat{G}_\alpha \sigma_{\alpha} \mper
  \]
  Therefore, we may apply Cauchy-Schwartz inequality to obtain
  \begin{align*}
    \rho(p(\mu)) & =  \sum_{\alpha}
                   \hat{F}_\alpha \hat{G}_\alpha \sigma_{\alpha} \mcom \\
                 & \leq \left(\sum_{\alpha} \hat{F}^2_\alpha \sigma_{\alpha}^{2}
                   \right)^{\frac{1}{2}} \left(\sum_{\alpha} \hat{G}^2_\alpha
                   \right)^{\frac{1}{2}}  \mcom\\
                 & = \left(\sum_{\alpha \neq 0} \hat{F}^2_\alpha \sigma_{\alpha}^{2}
                   \right)^{\frac{1}{2}} \mcom
  \end{align*}
  where in the last step we have used both that $\hat{F}_0 = 0$ and $\sum_{\alpha}
  \hat{G}_{\alpha}^2 = 1$. Let $\rho = \rho(\mu)$ and note that for all $\alpha \in
  L_\mu$ we have $\sigma_{\alpha} \leq \rho$. Further, for non-zero $\alpha \notin
  L_\mu$ we have $\sigma_\alpha \leq \rho^2$. Thus, we can split the above sum to
  obtain:
  \begin{align*}
    & \leq \left( \rho^2 \sum_{\alpha \in L_\mu} \hat{F}^2_\alpha + \rho^4
      \sum_{\alpha \notin L_\mu} \hat{F}^2_\alpha \right)^{\frac{1}{2}} \\
    & \leq \rho \left(\Lin_\mu(p)+ (1-\Lin_\mu(p)) \rho^2\right)^{\frac{1}{2}} \\
    & \leq \rho \left( 1- (1-\Lin_\mu(p))(1-\rho^2) \right)^{\frac{1}{2}} \\
    & \leq \rho \left( 1- \frac{1}{2}(1-\Lin_\mu(p))(1-\rho^2) \right) \mper
  \end{align*}
\end{proof}

\subsection{Hypercontractivity and the Berry-Esseen Theorem}
In light of \pref{lem:cordecaysym}, correlation will always decay under
application of an operation $p$ as long as $\Lin_{\mu}(p)$ is bounded away from
one. The main challenge for proving that correlation decays to zero under
repeated application of polymorphisms is in controlling the linearity after each
application. The intuition for our proof is as follows: If a polymorphism has
linearity very close to one, then it is close to a sum of independent random
variables, and so the Berry-Esseen Theorem applies to show that $p(\mu_1)$ is
nearly Gaussian. However this should be impossible since $p(\mu_1)$ only takes
$q$ distinct values. The version of the Berry-Esseen Theorem that we will use can be found as Corollary 59 in Chapter 11 of \cite{OD14}.

\begin{theorem}[Berry-Esseen \cite{OD14}]\label{thm:berryesseen}
Let $X_1,\dots,X_n$ and $Y_1,\dots,Y_n$ be independent random variables with matching first and second moments; i.e. $E[X_i^k] = E[Y_i^k]$ for $i\in [n]$ and $k \in {1,2}$. Let $S_X = \sum_i X_i$ and $S_Y = \sum_i Y_i$. Then for any $\psi:\R \to \R$ which is $c$-Lipschitz
\[
|\E[\psi(S_X)] - \E[\psi(S_Y)]| \leq O(c)\cdot \sum_i E[X_i^3] + E[Y_i^3]
\]
\end{theorem}

 Note that since the error term in the Berry-Esseen Theorem depends on the $L^3$-
norm of the random variables, we must control the $L^3$-norms of singular vectors
of $T_\mu$. Our main tool to this end will be hypercontractivity of the noise
operator $T_\rho$. We state here a special case of the general hypercontractivity
theorem which will suffice for our purposes.

\begin{theorem}[Hypercontractivity \cite{OD14}]\label{thm:hypercontractivity}
  Let $\pi$ be a probability distribution on $[q]$ where each outcome
  has probability at least $\lambda$. Let $f\in L^2([q]^k,\pi^k)$ for some $k \in \N$. Then for any $0 \leq \rho \leq \frac{1}{\sqrt{2}}\lambda^{\frac{1}{6}}$
  \[
  \norm{T_\rho f}_3 \leq \norm{f}_2
  \]
  where $T_\rho$ is the noise operator.
\end{theorem}

Our first task will be to relate $L^3$ norms under the conditional expectation
operator $T_\mu$ to those under the noise operator $T_\rho$. However, these
operators do not map between the same spaces. To fix this, we instead consider the
operator $M_\mu = T_\mu T_\mu^*$ so that $M:L^2(\Omega_1) \to L^2(\Omega_1)$. The
following simple lemma gives a hypercontractivity result for $M_\mu$.

\begin{lemma}\label{lem:hypercontractivity}
  Let $M_\mu = T_\mu T_\mu^*$ and let $\lambda$ be the minimum probability of an
  atom in the marginal $\mu_1$ of $\mu$. Let $s$ be a number such that
  $\rho(\mu)^{2s} \leq \frac{1}{\sqrt{2}}\lambda^{\frac{1}{6}}$.  For any $k \in \N$, if  let $f\in
  L^2([q]^k,\mu_1^k)$ then
  \[
  \norm{M_{\mu^k}^s f}_3 \leq \norm{f}_2
  \]
\end{lemma}
\begin{proof}
  Observe first that $M_\mu$ is a self-adjoint linear operator and let
  $\{\sigma_i^2 | i \in [q]\}$ be the eigenvalues and let $\{\phi_i| i \in [q]\}$ denote
  the eigenfunctions.. Further $M_{\mu^k} = M_\mu^{\otimes k}$ is a
  self-adjoint operator with eigenvalues $\{\sigma_\alpha^2| \alpha \in [q]^k\}$ and
  eigenfunctions $\{\phi_\alpha|\alpha \in [q]^k\}$.

  Let $\rho = \rho(\mu)^{2s}$ and define
  \[
  \eta_\alpha = \frac{\sigma_\alpha^{2s}}{\rho^{|\alpha|}}
  \]
  Observe that $\eta_0 = 1$. Next, because $\sigma_j \leq \rho(\mu)$ for all $j \neq 0$ and $\sigma_0 = 1$ we have
  \[
  \eta_\alpha = \frac{\prod_i \sigma_{\alpha_i}^{2s}}{\rho^{|\alpha|}}
  \leq \frac{\rho(\mu)^{2s|\alpha|}}{\rho^{|\alpha|}} = 1
  \]
  for all $\alpha\neq 0$. Now define an operator $A$ by setting $A\phi_\alpha = \eta_\alpha
  \phi_\alpha$ for each $\alpha \in [q]^k$. Since $\{\phi_\alpha\}_{\alpha \in [q]^k}$ form a basis,
  this uniquely defines $A$. In particular $A$ is a self-adjoint
  operator with eigenvalues $\eta_\alpha \leq 1$ and corresponding eigenfunctions
  $\phi_\alpha$ for $\alpha \in [q]^k$. Note that by construction
  \[
  T_\rho A \phi_\alpha = \sigma_\alpha^{2s} \phi_\alpha = M_{\mu^k}^s \phi_\alpha \mcom
  \]
  That is, $T_\rho A$ agrees with $M_{\mu^k}^s$ on the entire $\phi_\alpha$ basis, and so by linearity $T_\rho A = M_{\mu^k}^s$. Now we compute
  \[
  \norm{M_{\mu^k}^s f}_3 = \norm{T_\rho A f}_3 \leq \norm{Af}_2 \mcom
  \]
  where the final inequality follows by applying
  \pref{thm:hypercontractivity} to the function $Af$. Next since every eigenvalue of $A$ is at most 1 we have
  \[
  \norm{Af}_2 \leq \norm{f}_2 \mper
  \]
  Plugging this into the previous inequality completes the proof.
\end{proof}

We now apply the hypercontractivity theorem in order to control the $L^3$-norms
of singular vectors of $T_{p(\mu)}$ for some operation $p$. The following lemma
gives a trade-off between the $L^3$-norm of a singular vector of $T_{p(\mu)}$ and
the magnitude of the corresponding singular value. Further, the trade-off depends
only on properties of the initial distribution $\mu$. This in turn implies that
singular vectors with high $L^3$-norm have low singular values, and thus cannot
contribute much to the correlation $\rho(p(\mu))$.
\begin{lemma} \label{lem:threenorm}
  Let $\lambda$ be the minimum probability of an atom in the marginal $\mu_1$ of
  $\mu$. Let $s$ be a number such that $\rho(\mu)^{2s} \leq \frac{1}
  {\sqrt{2}}\lambda^{\frac{1}{6}}$. Let $p:[q]^k \to [q]$ and $\phi_i\in
  L^2(\Omega_1)$ be a singular vector of $T_{p(\mu)}$ with singular value
  $\sigma_i$ for $i \neq 0$. Then
  \[
  \sigma_i^{2s} \norm{\phi_i}_3 \leq 1
  \]
\end{lemma}
\begin{proof}
  For any distribution $\pi$ we will use the notation $M_\pi = T_\pi T_\pi^*$. Since
  $\phi_i$ is a singular vector of $T_{p(\mu)}$ with singular value $\sigma_i$ we
  have
  \[
  \sigma_i^{2s}\norm{\phi_i}_3 = \norm{M^{s}_{p(\mu)} \phi_i}_3 \mper
  \]
  In terms of the dual norm, we can write
  \begin{align*}
    \norm{M^{s}_{p(\mu)} \phi_i}_3 & = \sup_{\norm{g}_\frac{3}{2} = 1} \E_{z\sim p(\mu_1)}[g(z)M^{s}_{p(\mu)}\phi_i(z)] \\
                                   & = \sup_{\norm{g}_\frac{3}{2} = 1} \E_{x\sim\mu_1^k}[g(p(x)) M^{s}_{\mu^k}\phi_i(p(x))] \\
                                   & \leq \sup_{\norm{g}_\frac{3}{2} = 1}
                                     \norm{g(p)}_\frac{3}{2}
                                     \norm{M^{s}_{\mu^k}(\phi_i \circ p)}_3 \\
                                   & = \norm{M^{s}_{\mu^k}(\phi_i \circ p)}_3 \\
  \end{align*}
  where in the second to last line we have applied H{\"o}lders inequality in the
  space $L^2([q]^k,\mu^k)$. Now note that by \pref{lem:hypercontractivity}\mcom
  \[
  \norm{M^{s}_{\mu^k}(\phi_i \circ p)}_3 \leq \norm{\phi_i \circ p}_2 = 1 \mper
  \]
  Combining the above inequalities completes the proof.
\end{proof}
Next we will need the following technical lemma regarding discrete distributions
taking only few values (e.g. functions $F \in \Sp(p)$ for some polymorphism $p$).
Informally, the lemma states that such distributions cannot be close to sums of
many i.i.d. random variables.
\begin{lemma} \label{lem:finitevsnormal}
  Let $F$ be a real-valued random variable taking only $q$ values. Let $X_i$ for $i
  \in [m]$ be i.i.d. real-valued random variables with $\E[X_i] = 0$, $\Var[X_i]
  \leq \frac{1}{m}$ and $\E[|X_i|^3] \leq B$. Let $S =
  \sum_i X_i$. Then for some absolute constant $c > 0$,
  \[
  \E[|F - S|] \geq \frac{1}{8q} - cBm \mper
  \]
\end{lemma}
\begin{proof}
  Let $\{a_j\}_{j \in [q]}$ be the set of real values that $F$ takes. Now define
  the function
  \[
  d(z) = \min_j |z - a_j|
  \]
  which simply measures the distance from $z$ to the nearest $a_j$. Let $v = \sum_i
  \Var[X_i]$ and let $Y_i$ be a Gaussian random variable with first two moments matching those of $X_i$. Note that $\sum_i Y_i$ is a Gaussian with mean zero and variance $v$. Further observe that $d$ is $1$-Lipschitz and so we may apply \pref{thm:berryesseen} to obtain
  \[
  |\E_{z\sim\mathcal{N}(0,v)}\left[d(z)\right] - \E\left[d(S)\right]| \leq cBm
  \]
  for some absolute constant $c$.

  Next observe that the set of intervals such that $d(z) \leq \delta$ has total length $2q\delta$. Thus, the normal distribution $\mathcal{N}(0,v)$ has probability mass at most $2q\delta$ on the region where $d(z) \leq \delta$. This implies that
  \[
  \E_{z \sim \mathcal{N}(0,v)} [d(z)] \geq \delta (1 - 2q\delta)
  \]
  Combining this with the previous inequality yields
  \[
  \E[d(S)] \geq \delta(1 - 2q\delta) - cBm
  \]

  Now let $H = F - S$. Since $F$ only takes the values $a_j$ we have
  \begin{align*}
    \E [d(S)] & = \E [d(F - H)] \\
              & = \E [\min_j |F - H - a_j|] \\
              & \leq \E[|a_j - H - a_j|] = \E[|H|] \\
  \end{align*}
  Thus we conclude that
  \[
  \E[|H|] \geq \delta(1 - 2q\delta) - cBm
  \]
  Setting $\delta = \frac{1}{4q}$ yields the desired result.
\end{proof}

\subsection{Correlation Decay}
Throughout the remainder of this section, whenever we refer to a joint probability distribution $\nu$ on $[q]\times [q]^m$ we will fix a singular value basis $\{ \phi_\alpha,\psi_\alpha \}$ for the operator $T_{\nu}^{\otimes k}$. Let $\nu_1$ and $\nu_2$ be the respective marginals of $\nu$. Whenever we have a function $F\in L^2([q]^k,\nu_1^k)$ we will write $\hat{F}_\alpha$ to denote the Fourier coefficients of $F$ with respect to the basis $\phi_\alpha$.

 Now we are ready to show that $\Lin_\mu(p)$ stays bounded away from one for operations $p$ that are transitive symmetric. First, we need the following
simple claim which asserts that corresponding degree one Fourier coefficients of $p$ are all equal.

\begin{claim} \label{claim:degonesym}
  Let $p:[q]^k \to [q]$ be a transitive symmetric operation and let $F \in \Sp(p)$.
  Let $\nu_1$ be any distribution on $[q]$ and fix a Fourier basis $\phi_\alpha$ for $L^2([q]^k,\nu_1^k)$.
  Let $\alpha,\beta$ be multi-indices with $|\alpha| = |\beta| = 1$. Further
  suppose $\alpha_i = \beta_j$ for the unique pair $i$ and $j$ such that $\alpha_i$
  and $\beta_j$ are non-zero. Then $\hat{F}_\alpha = \hat{F}_\beta$.
\end{claim}

\begin{proof}
  Let $\pi$ be a permutation such that $\pi(j) = i$ and $p(x) = p(\pi \circ x)$. Such a permutation always exists since $p$ is transitive symmetric. Now note
  \[
  \hat{F}_\alpha = \E [F(x)\phi_{\alpha_i}(x_i)] = \E [F(\pi \circ x)\phi_{\alpha_i}(x_i)]
  = \E [F(x)\phi_{\beta_j}(x_j)] = \hat{F}_\beta
  \]
\end{proof}

We turn now to the proof of the theorem. The main idea of the proof is to
repeatedly apply \pref{lem:cordecaysym} to reduce correlation while using
\pref{lem:threenorm} in conjunction with \pref{lem:finitevsnormal} to
control the linearity in each step.
\begin{theorem} \label{thm:fullcordecay}
  Let  $p_1\dots p_r$ be transitive symmetric polymorphisms with each $p_i : [q]^k
  \to [q]$. Let $\lambda$ be the minimum probability of an atom in the marginal
  $\mu_1$ of $\mu$. Then for any $\frac{1}{40c} > \eps > 0$ and $r =
  \Omega_q(\frac{\log \lambda}{\log \rho(\mu)}\log^2(\eps^{-1}))$
  \[
  \rho(p_1 \otimes p_2 \otimes \dots \otimes p_r(\mu)) \leq \eps
  \]
\end{theorem}
\begin{proof}
  For the analysis we will break $p_1 \dots p_r$ into consecutive segments of
  length $a$, where $a$ is a parameter that we will set later. Formally, let $K =
  k^{a}$ and let $P_t:[q]^K \to [q]$ be defined as
  \[
  P_t(x) = p_{at + 1} \otimes \dots \otimes p_{at}(x)
  \]
  Let $\mu^{(1)} = \mu$ and $\mu^{(t)} = P_1 \otimes \dots \otimes P_{t-1}(\mu)$. Observe
  that each $P_t$ is again a transitive symmetric polymorphism. Next we control
  $\Lin_{\mu^{(t)}}(P_t)$ only in terms of $\rho(\mu^{(t)})$ and properties of the initial
  distribution $\mu$.

  Let $F \in \Sp(P_t)$ so that $\E_{\mu^{(t)}}[F] = 0$ and $\norm{F}_2 = 1$. Let $\{ \phi_\alpha,\psi_\alpha \}$ be a singular value basis for $T_{\mu^{(t)}}^{\otimes K}$ and let $\hat{F}_\alpha$ be the Fourier coefficients of $F:[q]^K \to \R$ with respect to the $\phi_\alpha$ basis. Now define the
  function $l_i:[q] \to \R$ to be
  \[
  l_i(x) = \sum_{\substack{\alpha \in L_{\mu^{(t)}}\\ \alpha_i \neq 0}} \hat{F}_\alpha \phi_{\alpha_i}(x)
  \]
  In words, each $l_i$ is simply the linear part of $F$ corresponding to the $i$th
  coordinate. Claim \ref{claim:degonesym} implies that for any pair $i$ and $j$ the
  sums for $l_i$ and $l_j$ are equal term-by-term. Thus, $l_i(x) = l_j(x)$ for all
  $i,j$.  Now note that
  \[
  F(x) = \sum_i l_i(x_i) + H(x)
  \]
  where $H(x) = \sum_{\alpha \notin L_{\mu^{(t)}}}\hat{F}_\alpha \phi_\alpha(x)$ is the
  non-linear part of $F$. Since $\norm{F}_2 = 1$ and the $l_i$ are all equal and
  orthogonal, it follows that $\norm{l_i}^2_2 \leq \frac{1}{K}$ for all $i$. This
  further implies that for any $\alpha \in L_{\mu^{(t)}}$ the coefficient
  $\hat{F}_\alpha \leq \frac{1}{\sqrt{K}}$.

  Thus, the $l_i(x_i)$ are i.i.d random variables with mean zero and variance at most $\frac{1}{K}$. Further note that
  \[
  \norm{l_i}_3 \leq \frac{1}{\sqrt{K}}\sum_{\substack{\alpha \in L_{\mu^{(t)}}\\ \alpha_i \neq 0}}
  \norm{\phi_{\alpha_i}}_3 \mper
  \]
  Let $s = \Omega\left(\frac{\log\lambda}{\log\rho}\right)$. Now we apply \pref{lem:threenorm} with $p = P_1 \otimes \dots \otimes
  P_{t-1}$ to obtain
  \[
  \frac{1}{\sqrt{K}}\sum_{\substack{\alpha \in L_{\mu^{(t)}}\\ \alpha_i \neq 0}} \norm{\phi_{\alpha_i}}_3
  \leq \frac{1}{\sqrt{K}}\sum_{\substack{\alpha \in L_{\mu^{(t)}} \\ \alpha_i \neq 0}}\sigma_\alpha^{-2s}
  \leq \frac{q}{\sqrt{K}}\rho(\mu^{(t)})^{-2s} \mcom
  \]
  where in the last inequality we have use the fact that $\sigma_\alpha \leq
  \rho(\mu^{(t)})$ for $\alpha \neq 0$, and that the sum has $q$ terms.
  Therefore, since $F(x)$ is a random variable taking only $[q]$ different values
  we may apply \pref{lem:finitevsnormal} to obtain
  \[
  \E[|H(x)|] \geq \frac{1}{8q} - c\norm{l_i}_3^3K \geq \frac{1}{8q} -
  cq^3K^{-\frac{1}{2}}\rho(\mu^{(t)})^{-6s} \mper
  \]
  Note that setting $a = 4\log(cq^3\eps^{-6s})$ yields
  \[
  K^{-\frac{1}{2}} = k^{-\frac{1}{2}a} \leq (cq^3\eps^{-6s})^{-2} \mper
  \]
  Thus as long as $\rho(\mu^{(t)}) \geq \eps$ we have
  \[
  \E[|H(x)|] \geq \frac{1}{8q} - (cq^3\eps^{-6s})^{-1} > \frac{1}{10q}
  \]
  where the last inequality follows from $\eps < \frac{1}{40c}$.
  So, letting $\delta = \E[H(x)^2] \geq \E[|H(x)|]^2 = \frac{1}{100q^2}$ we obtain $\Lin_{\mu^{(t)}}(P_t) \leq 1 - \delta$.

  Now we apply \pref{lem:cordecaysym} to $P_t$ to obtain
  \begin{align*}
    \rho(\mu^{(t+1)}) & \leq \rho(\mu^{(t)})\left(1 - \frac{1}{2}(1 - \Lin(\mu^{(t)}))(1 - \rho(\mu^{(t)})^2)\right) \\
                      & \leq \rho(\mu^{(t)})\left(1 - \frac{\delta}{2}(1 - \rho(\mu^{(t)})^2)\right) \\
                      & = \left(1 - \frac{\delta}{2}\right)\rho(\mu^{(t)}) + \frac{\delta}{2}\rho(\mu^{(t)})^3
  \end{align*}
  Solving this recurrence shows that $\rho(\mu^{(t)}) \leq \eps$ after $t = O_q(\log(\eps^{-1}))$ steps. Since each step corresponds to $a$ applications of polymorphisms, we get that the total number of applications required is
  \[
  r = \Omega_q(a\log(\eps^{-1})) = \Omega_q(s\log^2(\eps^{-1})).
  \]
\end{proof}

We now use the connection given in \pref{lem:corr-stat-distance} between correlation and statistical distance to prove \pref{thm:corr-decay}.
\begin{proof}[Proof of \pref{thm:corr-decay}]
  Let $P = p_1 \otimes p_2 \otimes \dots \otimes p_r$. The proof that $\norm{P(\mu) - P(\mu^\times)}_1 \leq \eta$ is by a hybrid argument. We define intermediate distributions $\pi^{(i)}$ by drawing a sample $X$ from $\mu$ and then independently re-sampling each of the coordinates $j > i$ from their respective marginals $\mu_j$. Let $Y_j$ be the random variable corresponding to an independent sample from the marginal $\mu_j$. We will use the notation $\pi^{(i)} = (X_1,\dots,X_i,Y_{i+1},\dots,Y_n)$.

  Note that since the $Y_j$ are independent of each other and of the $X_l$ we have
  \[
  \norm{P(\pi^{(i-1)}) - P(\pi^{(i)})}_1 = \norm{\Big(P(X_1),\dots ,P(X_{i-1}),P(Y_i)\Big) - \Big(P(X_1),\dots,P(X_i)\Big)}_1
  \]
  Let $\eps = \frac{\eta}{qn}$. By \pref{thm:fullcordecay} we have
  $\rho\Big(P(X_i),(P(X_1),\dots,P(X_{i-1}))\Big) \leq \eps$. Further, since $Y_i$
  is independent of $X_1,\dots,X_{i-1}$ we may apply \pref{lem:corr-stat-distance} to obtain
  \[
  \norm{\Big(P(X_1),\dots ,P(X_{i-1}),P(Y_i)\Big) - \Big(P(X_1),\dots,P(X_i)\Big)}_1 \leq q\eps
  \]
  Now since $\pi^{(0)} = \mu^\times$ and $\pi^{(n)} = \mu$ we have by the triangle inequality
  \[
  \norm{P(\mu) - P(\mu^\times)}_1 \leq \sum_{i = 0}^{n-1} \norm{P(\pi^{(i)}) - P(\pi^{(i + 1)})}_1 \leq nq\eps = \eta
  \]
\end{proof}
 \label{sec:corr-decay}

\section{Approximate Polymorphisms for CSPs} \label{sec:maxcsps}
\subsection{Background} \label{sec:cspdefs}

\paragraph{Constraint Satisfaction Problems}

Fix an alphabet $[q]$.  A CSP $\Lambda$ over $[q]$ is given by a
family of payoff functions $\Lambda = \{c:[q]^k \to [-1,1]\}$.
An instance of \maxl consists of a set of variables $\cV = \{X_1,\ldots,
X_n\}$ and a set of constraints $\cC = \{C_1,\ldots,C_m\}$ where each $C_i(X) =
c(X_{i_1},X_{i_2},\ldots, X_{i_k})$ for some $c \in \Lambda$.  The set
$\cC$ is equipped with a probability distribution $w:\cC \to \R^+$.

The goal is to find an assignment $x \in [q]^{\cV}$ that maximizes
$$\val_\inst(x) \defeq  \sum_{C \in \cC}w(C)C(x) \mper$$
If the probability distribution $w$ is clear from the context, we will
write
$$\val_\inst(x) = \E_{C \sim \cC}\left[C(x)\right] \mper$$
\begin{remark}
  A natural special case of the above definition is the set of
  unweighted constraint satisfaction problems wherein the weight
  function $w: \cC \to \R_+$ is uniform, i.e.,
$$\val_{\inst}(x) = \frac{1}{|\cC|} \sum_{C\in \cC} C(x) \mper$$
In terms of approximability, the unweighted case is no easier than the
weighted version since for every constant $\eta>0$, there is a polynomial
time reduction from weighted version to its unweighted counterpart with a
loss of at most $\eta$ in the approximation.
\end{remark}

\paragraph{Dictatorship tests}

\begin{definition}
Fix constants $-1 \leq s \leq c \leq 1$, an integer $R \in \N$ and a family of
functions $\cF = \{A:[q]^R \to [q]\}$.  A $(c,s)$-dictatorship test for
\maxl against the family $\cF$ consists of an instance $\inst$
of \maxl with the following properties:  The variables of
$\inst$ are identified with $[q]^R$ and therefore an assignment to
$\inst$ is given by $p: [q]^R \to [q]$.
  \begin{itemize}
  \item(completeness) For every $i \in [R]$, the $i^{th}$ dictator assignment $p_i :
    [q]^R \to [q]$ given by $p_i(x) = x^{(i)}$ satisfies $\inst(p_i)\geq c$.

  \item(soundness) For every function $p \in \cF$, we have $\inst(p) \leq s$
 \end{itemize}

\end{definition}

  The intimate relation between polymorphisms and dictatorship testing was
observed in \cite{KunS09}, here we spell out the connection to
approximate polymorphisms.  First, for a function family $\cF = \{ p: [q]^R \to [q]\}$, we will say an
approximate polymorphism $\cP$ is {\it supported on} $\cF$ if
probability distribution $\cP$ is supported only on functions within
the family $\cF$.

\begin{theorem} \label{thm:poly-to-dict-equivalence}
 Given a CSP $\Lambda$, a natural number $R \in \N$ and a finite family of
 functions $\cF = \{p:[q]^R \to [q]\}$ closed under permutation of inputs, the following are equivalent:
\begin{itemize}
\item \maxl does not admit a $(c,s)$-approximate polymorphism
  supported on $\cF$.
\item There exists a $(c,s)$-dictatorship test for \maxl
  against the family $\cF$.
\end{itemize}
\end{theorem}

\begin{proof}
  Suppose $\maxl$ does not admit $(c,s)$-approximate
  polymorphism supported in $\cF$.
  In other words, for every probability distribution $\cP$ over $\cF$,
  there exists an instance $\inst$ and $R$ solutions
  $X^{(1)},\ldots,X^{(R)}$ of average value $c$ such that
  the value of $\E_{p \in \cP}\inst(p(X^{(1)},\ldots,X^{(R)})) \leq
  s$
  %


  Notice that the set of probability distributions $\cP$ over
  $\cF$ forms a convex set.  Furthermore, the set of all instances $\inst$ of
  \maxl along with $R$ assignments $X^{(1)},\ldots,X^{(R)}$ with average
  value at least  $c$, form a convex set.  Specifically, given two instances
  $\inst_1, \inst_2$ and any $\theta \in [0,1]$, one can construct the instance
  $\theta \inst_1 + (1-\theta) \inst_2$ by taking a disjoint union of the two
  instances weighted appropriately.  Since the set of all probability
  distributions $\cP$ over $\cF$ is a compact convex set, by
  min-max theorem there exists a single instance $\inst$ and a
  set of $R$ solutions for it, that serve as a counterexample against
  every probability distribution $\cP$ over the family $\cF$.
  %

  We will use the instance $\inst$ and the set of solutions
  $X^{(1)},\ldots,X^{(R)}$ to create the following dictatorship
  test.

  \begin{itemize}
  \item Sample a random constraint $C$ from $\inst$.  Let $C$ be
    a constraint on variables $v_{i_1},\ldots,v_{i_k}$ in
    $\inst$.
  \item Pick a random permutation $\pi :[R] \to [R]$ and set
    $y_{j} = (X_{i_j}^{(\pi(1))},\ldots,X_{i_j}^{(\pi(R))})$ for
    $j \in [R]$.
  \item Test the constraint $C$ on $p(y_{i_1}),\ldots,p(y_{i_k})$.
  \end{itemize}

  Suppose $p(y) = y^{(i)}$ for some $i \in [R]$.  In this case,
  it is easy to see that over the random choice of constraint
  $C$ and permutation $\pi$, the expected probability of success
  of the dictatorship test is exactly equal to the average value
  of the solutions $X^{(1)},\ldots,X^{(R)}$.

  Let $p:[q]^R \to [q]$ denote any function in the family $\cF$.
  If probability of success of $p$ is greater
  than $s$ then there exist a
  permutation $\pi:[R] \to [R]$ for which,
  $$ \inst(p(X^{(\pi(1))},\ldots,X^{(\pi(R))})) \geq
  s \mper$$
  This is a contradiction, since $p \circ \pi$ would also be a function in
  the family $\cF$.  Therefore, for every function $p \in \cF$, its value on
  the dictatorship test is at most $s$.  This completes the proof of one direction of the
  implication.

  Now, we will show the other direction, namely that if there exists a
  $(c,s)$-dictatorship test then there is no $(c,s)$-approximate polymorphism.
  Conversly, suppose there is a $(c,s)$-dictatorship test against the
  family $\cF$.  Consider the instance $\inst$ given by the dictatorship
  test and the family of $R$ assignments corresponding to the dictator
  functions, i.e., $p_i(x) = x^{(i)}$.  By the completeness of the
  dictatorship test, for each of the dictator assignments $p_i$, we
  will have $\inst(p_i) \geq c$.  On the other hand, by the soundness of
  the dictatorship test, for each function $p$ in the family $\cF$,
  $\inst(p) \leq s$.  This implies that there does not exist a
  $(c,s)$-approximate polymorphism supported on the family $\cF$.

\end{proof}

\paragraph{Unique Games Conjecture}
For the sake of completeness, we recall the unique games conjecture here.
\begin{definition}
  An instance of Unique Games represented as $\Gamma =
  (\mathcal{A}\cup \mathcal{B},E,\Pi,[R])$, consists of a bipartite
  graph over node sets $\mathcal{A}$,$\mathcal{B}$ with the edges $E$
  between them.  Also part of the instance is a set of labels $[R] =
  \{1,\ldots,R\}$, a a set of permutations $\pi_{ab} : [R]
  \rightarrow [R]$ for each edge $e = (a,b) \in E$.  An assignment
  $\Lambda$ of labels to vertices is said to satisfy an edge $e =
  (a,b)$, if $\pi_{ab}(\Lambda(a)) = \Lambda(b)$.  The objective is to
  find an assignment $\Lambda$ that satisfies the maximum number
  of edges.
\end{definition}
The unique games conjecture of Khot \cite{Khot02a} asserts that the
unique games CSP is hard to approximate in the following sense.
\begin{conjecture}(Unique Games Conjecture \cite{Khot02a}) For
  all constants $\delta > 0$, there exists large enough constant $R$
  such that given a bipartite unique games instance $\Gamma =
  (\mathcal{A} \cup \mathcal{B},E, \Pi = \{\pi_{e}:[R]\rightarrow [R]
  ~ : ~ e \in E\},[R])$ with number of labels $R$, it is NP-hard to
  distinguish between the following two cases:
  \begin{itemize}
    \itemsep=0ex
  \item {\sf $(1-\delta)$-satisfiable instances:} There exists an
    assignment $\Lambda$ of labels that satisfies a
    $(1-\delta)$-fractional of all edges.
  \item {\sf Instances that are not $\delta$-satisfiable:}  No
    assignment satisfies more than a $\delta$-fraction of the
    constraints $\Pi$.
  \end{itemize}
\end{conjecture}

\subsection{Quasirandom functions}

The function family of interest in this work are those with no
influential coordinates.  To make the definition precise, we begin by
recalling a few analytic notions.
\paragraph{Low Degree Influences}

%
%
%
Fix a probability distribution $\mu$ over $[q]$.  Let $\{\chi_0,\ldots,\chi_{q-1}\}$
be an orthonormal basis for the vector space $L_2([q],\mu)$.
Without loss of generality, we can fix a basis such that $\chi_0 = 1$.
Given a function $f:[q]^k \to \R$, we can write $f$ as,
$$f = \sum_{\sigma \in \N^k} \hat{f}_\sigma \chi_\sigma \mcom$$
where $\chi_\sigma(x) \defeq \prod_{j = 1}^k \chi_{\sigma_j}(x_j)$.
Define the degree $d$ influence of the $i^{th}$ coordinate of $f$
under the probability distribution $\mu$ as,
$$ \Inf^{< d}_{i,\mu}(f) \defeq \sum_{ \sigma \in \N^k, \sigma_i \neq 0, |\sigma|
  < d}
\hat{f}_\sigma^2 \mper$$
More generally, for a vector valued function $f: [q]^k \to \R^D$, we will set
$$\Inf_{i,\mu}^{< d}(f) \defeq \sum_{j \in[D]} \Inf_{i,\mu}^{<  d}(f_j) \mper$$
A useful property of low-degree influences is that their sum is
bounded as expounded in the following lemma.

\begin{lemma} \label{lem:suminfluences}
For every function $f: [q]^k \to \R$ and every probability distribution
$\mu$ over $[q]$,
$$\sum_{i \in [k]} \Inf_{i,\mu}^{< d} (f) \leq d \cdot \Var_{\mu}(f) \mcom$$
where $\Var_{\mu}(f)$ denotes the variance of $f$ under the probability
distribution $\mu$.
\end{lemma}

Given an operation $p:[q]^k \to [q]$, we will associate a function
$\tilde{p}:[q]^k \to \ssimp_q$ given by $\tilde{p}(x) = e_{p(x)}$ where in
$e_{p(x)} \in \ssimp_q$ is the basis vector along $p(x)^{th}$
coordinate.  We will abuse notation and use $p$ to denote both the
$[q]$-valued function $p$ and the corresponding real-vector valued
function $\tilde{p}$.  For example, $\Inf_{i,\mu}(p)$ will denote the influence
of the $i^{th}$ coordinate on $\tilde{p}$.

\begin{definition}
 An approximate polymorphism $\cP$ is $(\tau,d)$-quasirandom if for every
 probability distribution $\mu \in \ssimp_q$,
$$\E_{p\in \cP}\left[ \max_{i } \Inf_{i,\mu}(p) \right] \leq \tau$$
\end{definition}

%
%
%
%
The following lemma shows that transitive symmetries imply
quasirandomness.
\begin{lemma}
  An operation $p:[q]^k \to [q]$ that has a transitive symmetry
  is $(\frac{qd}{k}, d)$-quasirandom for all $d \in \N$.
\end{lemma}
\begin{proof}
  The lemma is a consequence of two facts, first the sum of degree $d$
  influences of a function $p:[q]^k \to [q]$ is always at most $qd$.
  Second, if the operation $p$ admits a transitive symmetry, then
  the degree $d$ influences of each coordinate is the same.
\end{proof}

\paragraph{Noise Operator}
For a function $p : [q]^k \to \R$ and $\rho \in [0,1]$, define $T_{\rho} p$ as
\begin{displaymath} T_{\rho} p (x) = \E_{y \sim_{\rho} x}[p(y)]\end{displaymath}
where $y \sim_{\rho} x$ denotes the following distribution,
$$ y_i = \begin{cases} x_i & \text{ with probability } \rho \\ \text{
    independent sample from } \mu & \text{ with probability } 1-\rho \mper
\end{cases}$$
The multilinear polynomial associated with $T_{\rho} p$ is given by
\begin{displaymath}
  T_{\rho} p =  \sum_{\sigma} \hat{p}_{\sigma}  \rho^{|\sigma|} \chi_{\sigma}\mper
\end{displaymath}

\paragraph{Approximation Thresholds}

For each $\tau > 0$ and $ d \in \N$, let $\cF_{\tau,d}$ denote the family of all
$(\tau,d)$-quasirandom functions.

\begin{definition}
  Given a CSP $\Lambda$, and a constant $c \in [-1,1]$ define
$$s_{\Lambda}(c) \defeq \sup \left\{ s | \forall \tau > 0, d \in \N,
  \exists \text{ a } (\tau,d)\text{-quasirandom }  (c , s)\text{-approximate polymorphism for \maxl} \right\}$$
Moreover, let
$$\alpha_{\Lambda} \defeq \inf_{c \geq 0} \frac{s_{\Lambda}(c)}{c}$$
\end{definition}
The following observations are immediate consequences of the
definition of $s_{\Lambda}$.
\begin{observation}\label{obs:approxpoly}
The map $s_{\Lambda}: [-1,1] \to [-1,1] $ is monotonically increasing and $s_{\Lambda}(c+\e) \leq s_{\Lambda}(c)+\e$ for every $c,\e$ such that $c,c+\e \in (-1,1)$.
\end{observation}

\subsection{Rounding Semidefinite Programs via Polymorphisms}
\label{sec:roundingcsps}
In this section, we will restate the soundness analysis of
dictatorship tests in \cite{Raghavendra08} in terms of approximate
polymorphisms.  Specifically, we will construct a rounding scheme for
the {\sc BasicSDP} relaxation using approximate polymorphisms, and
thereby prove \prettyref{thm:sdpgap}.


\paragraph{Basic SDP relaxation}
Given a $\Lambda$-CSP instance $\inst = (\cV, \cC)$, the goal of {\sc
  BasicSDP} relaxation is to find a collection
of vectors $\set{\vec b_{v,a}}_{v\in \cV, a\in [q]}$ in a sufficiently
high dimensional space and a collection
$\set{\mu_C}_{C\in \supp(\cC)}$ of
distributions over local assignments.  For each payoff $C \in \cC$, the distribution
$\mu_C$ is a distribution over $[q]^{\cV(C)}$ corresponding to
assignments for the variables $\cV(C)$.  We will write $\Pr_{x \in
  \mu_{C}} \Set{E}$ to denote the probability of an event $E$ with
under the distribution $\mu_C$.
\vspace{10pt}
\begin{mybox}
  {\sc BasicSDP} Relaxation
  \begin{align*} \label{sdpgen}
    \maximize \quad & \E_{C\sim \cC} \E_{x\sim\mu_C} C(x) \qquad
                  \qquad \qquad \tag{ {\sc Basic SDP}} \nonumber  \\
    \subjectto\quad %
                & \iprod{\qV{v,i},\qV{v',j}} = %
                  \Pr_{x \sim \mu_C}\Set{\vbig x_v=i, x_{v'}=j}\\ %
                & \qquad \qquad  \text{($C\in \supp(\cC)$, $\ v,v'\in \cV(C)$, $\
                  i,j\in[q]$)}\mper\\ 
                & \mu_C \in \ssimp([q]^{\cV(C)}) \qquad \qquad \qquad \forall C \in \supp(\cC)\nonumber
  \end{align*}
\end{mybox}

We claim that the above optimization problem can be solved in
polynomial time. To show this claim, let
us introduce additional real-valued variables $\mu_{C,x}$ for
$C\in\supp(\cC)$ and
$x\in[q]^{V(C)}$. We add the constraints $\mu_{C,x}\ge 0$ and
$\sum_{x\in[q]^{V(C)}}\mu_{C,x}=1$. We can now make the following
substitutions to eliminate the distributions $\mu_C$,
\begin{align*}
  \E_{x\sim\mu_C} C(x)  = \sum_{x\in[q]^{V(C)}} C(x)\mu_{C,x}\mcom
  \qquad \qquad
  \Pr_{x \sim \mu_C}\Set{\vbig x_i=a}  = \sum_{\substack{x\in
  [m]^{V(C)}\\x_i=a}} \mu_{C,x} \mcom \\
  \Pr_{x \sim \mu_C}\Set{\vbig x_i=a, x_j=b}  = \sum_{\substack{x\in
  [m]^{V(C)}\\x_i=a,x_j=b}} \mu_{C,x}
  \mper
\end{align*}
After substituting the distributions $\mu_C$ by the scalar variables
$\mu_{C,x}$, it is clear that an optimal solution to the relaxation of
$\cC$ can be computed in time $\poly(m^k, |\supp(\cC)|)$.
In the rest of the section, we will show how to obtain an
rounding scheme for the above SDP using an
approximate polymorphism $\cP$.

The overall idea behind the rounding scheme is as follows.
Let $(\vec V, \vec \mu)$ be a feasible solution to the {\sc BasicSDP}
where $\vec V$ consists of the vectors and $\vec \mu$ consists of the
associated local distributions.

Let $\cP$ be an $(c,s)$-approximate polymorphism.  If the
polymorphism $\cP$ is given as input integral assignments whose average value is equal
to $c$, it would output an integral assignment of
expected value
$s$.  This would be a rounding for {\sc BasicSDP} relaxation
certifying that the on instances with SDP value at least $c$, the
optimum is at least $s$.  However,
in general, we do not have access to integral solutions of value
$c$ and there might not exist any.
The idea is to give  as input to $\cP$ real-valued assignments such
that they have value $c$, and the polymorphism
$\cP$ cannot distinguish these real valued assignments from integral
assignments.  Specifically, these real valued assignments are gaussian
random variables obtained by taking random projections of the SDP
vectors.

 The following is the formal description of the
rounding procedure $\round_{\cP}$.
\vspace{10pt}
\begin{mybox}
  $\round_{\cP}$ Scheme\\\\
  \textbf{Input:} A $\Lambda$-CSP instance  $\inst = (\cV,\cC)$
  with $m$ variables and an SDP
  solution $\{ \qV{v,i}\}, \{\mu_C\}$ with value at least $c$. A
  $(\tau,d)$-quasirandom $(c-\eta,s)$-approximate polymorphism $\cP$. \\

  \textbf{Truncation Function} Let $\ssimptrunc : \R^q \to \ssimp_q$ be a Lipschitz-continous
  function such that for all $\mv{x} \in \ssimp_q$, $\ssimptrunc(\mv{x})
  = \mv{x}$. \\ 
  \textbf{Rounding Scheme:}
  Fix $\e = \eta/10k$ where $k$ is arity of $\Lambda$.

  \begin{itemize}
  \item Sample $R$ vectors $\zeta^{(1)},\ldots,\zeta^{(R)}$ with each coordinate
  being i.i.d normal random variable.

  \item Sample an operation $p \sim  \cP$.

  \end{itemize}

  For each $v \in \cV$ do
  \begin{itemize}

  \item For all $1 \leq  j \leq R $ and $c \in [q]$,
    compute the projection $g_{v,c}^{(j)}$ of the vector
    $\qV{v,c}$ as follows:
    \begin{align*}
      {g}_{v,c}^{(j)}  = \iprod{\vzero , \qV{v,c}} +
      \Big[\iprod{(\qV{v,c} - (\iprod{\vzero , \qV{v,c}}) \vzero) ,
      \zeta^{(j)}}\Big]
    \end{align*}

  \item     Let $\mpl{H}$ denote the multilinear polynomial
    corresponding to the function $T_{1-\e}p : \ssimp_q^k \to \ssimp_q$.
    Evaluate the function $\mpl{H} = T_{1-\eps} p$ with
    $g_{v,c}^{(j)}$ as
    inputs.  In other words, compute $\mv{z}_v =
    (z_{v,1},\ldots,z_{v,q})$ as follows:
    $$\mv{z}_v = \mpl{H}(\mv{g_v})$$
  \item     Round $\mv{z}_v$ to $\mv{z}_v^{*} \in
    \ssimp_q$ by using a Lipschitz-continous truncation function
    $\ssimptrunc : \R^q \to \ssimp_q$, i.e., $ \mv{z}_v^* = \ssimptrunc(\mv{z}_v) \mper$
  \item   Independently for each $v \in \cV$, assign a value $j \in [q]$
    sampled from the distribution $z^{*}_{v} \in \ssimp_q$.
  \end{itemize}
\end{mybox}
\vspace{10pt}

Now we will analyze the performance of the above rounding scheme.
First, we observe the following fact about approximate polymorphisms.
\begin{lemma} \label{lem:averagevalue}
Fix an instance $\inst$ of \maxl and a distribution $\Theta$ over assignments to
$\inst$ with $\E_{X \sim \Theta} [\inst(X)] \geq c$.
Suppose $\cP$ is a $(c -\eta,s)$-approximate polymorphism for \maxl for
some $\eta > 0$, then we have
$$ \E_{p \in \cP} \E_{X^{(1)},\ldots, X^{(R)} \sim \Theta} \left[
  \inst(p(X^{(1)},\ldots,X^{(R)})) \right] \geq s \mper$$
where $X^{(1)},\ldots,X^{(R)}$ are sampled independently from $\Theta$.
\end{lemma}
\begin{proof}
Let $N \in \N$ be a large positive integer.  Let $\inst'$ be an instance
consisting of $N$ disjoint copies of $\inst$.  Define a distribution
$\Theta'$ of assignments to $\inst'$ consisting of $N$-i.i.d samples from $\Theta$.
While the distribution $\Theta$ only satisfies an average bound on
objective value $\E_{X \sim \Theta} [\inst(X)] \geq c$,  the distribution
$\Theta'$ will satisfy
$$\Pr_{Y \sim \Theta'}[\inst'(Y) \geq c-\eta] \geq 1-e^{-O(\eta N)} \mper$$
Hence, if we sample $Y^{(1)},\ldots,Y^{(R)} \sim \Theta'$ then the average value
of the assignments is at least $c-\eta$ with probability $1-e^{-O(\eta
  N)}$.  We are ready to finish the argument as shown below,
\begin{align*}
 \E_{p \in \cP} \E_{X^{(1)},\ldots, X^{(R)} \sim \Theta} \left[
  \inst(p(X^{(1)},\ldots,X^{(R)})) \right]  & =  \E_{p \in \cP} \E_{Y^{(1)},\ldots, Y^{(R)} \sim \Theta'} \left[
  \inst'(p(Y^{(1)},\ldots,Y^{(R)})) \right]  \\
 \quad& \qquad\quad  (\text{ linearity of expectation
                                    })\\
& \geq  \E_{p \in \cP} \E_{Y^{(1)},\ldots, Y^{(R)} \sim \Theta'} \left[
  \inst'(p(Y^{(1)},\ldots,Y^{(R)})) \suchthat \inst'(Y^{(i)})\geq c-\eta \right]
  - e^{-O(\eta N)} \\
& \geq s - e^{-O(\eta N)} \mper
\end{align*}
Taking limits as $N \to \infty$, the conclusion is immediate.
\end{proof}

To analyze the performance of the rounding scheme, we define a set of
ensembles of local integral random variables,
and global gaussian ensembles as follows.

\begin{definition}
  For every payoff $C \in \cC$ of size at most $k$, the local distribution
  $\mu_C$ is a distribution over $[q]^{\cV(C)}$.  In other words, the
  distribution $\mu_{C}$ is a distribution over assignments to the
  CSP variables in set $\cV(C)$.
  The corresponding {\it local integral ensemble } is a set of random
  variables $\mcl{L}_{C} = \{\mv{\ell}_{v_1},\ldots,\mv{\ell}_{v_k}\}$ each taking
  values in $\ssimp_q$.
\end{definition}
\begin{definition}
  The {\it global ensemble} $\erv{G} = \{\mrv{g}_{v} | v \in \cV
  ,j \in [q]\}$ are generated by setting $ \mrv{g}_{v} =
  \{g_{v,1},\ldots,g_{v,q}\} $ where
  $$ g_{v,j} \defeq \iprod{\vzero , \qV{v,j}} + \iprod{(\qV{v,j} -
    (\iprod{\vzero,
      \qV{v,j}}) \vzero), \zeta}$$
  and $\zeta$ is a normal Gaussian random vector of appropriate dimension.
\end{definition}
It is easy to see that the local and global integral ensembles have
matching moments up to degree two.
\begin{observation}
  For any set $C \in \cC$, the global ensemble $\erv{G}$
  matches the following moments of the local integral ensemble
  $\mcl{L}_{C}$
  \begin{align*}
    \E[g_{v,j}]  = \E[\ell_{v,j}] = \iprod{\vzero , \qV{v,j}}  & &
                                                                   \E[g_{v,j}^2]  = \E[\ell_{v,j}^2] = \iprod{\vzero , \qV{v,j}} \\
    \E[g_{v,j} g_{v,j^{\prime}}] = \E[\ell_{v,j}
    \ell_{v,j^{\prime}}] = 0  &  & \forall j \neq j^{\prime}, v \in
                                   \cV(C)
  \end{align*}
\end{observation}
We will appeal to the invariance principle of Mossel
\etal \cite{MosselOO05}, Mossel and Issakson \cite{IsakssonM09} to argue
that a $(\tau,d)$-quasirandom polymorphism cannot
distinguish between two distributions that have same first two
moments.
\begin{theorem}(Invariance Principle \cite{IsakssonM09}) \label{thm:invariance}
  Let $\Omega$ be a finite probability space with the least non-zero
  probability of an atom at least $\alpha \leq 1/2$.  Let $\mcl{L} =
  \{\ell_1,\ell_1,\ldots,\ell_{m}\}$ be an ensemble of random variables
  over $\Omega$.  Let $\erv{G} =
  \{g_1,\ldots,g_{m}\}$ be an ensemble of Gaussian random variables satisfying the following conditions:
  \begin{align*}
    \E[\ell_i] = \E[g_i]  & & \E[\ell_i^2] = \E[g_i^2] & & \E[\ell_i
                                                           \ell_j] = \E[g_i g_j] & & \forall i,j \in [m]
  \end{align*}
  Let $K = \log (1/\alpha)$.  Let $\mv{F} = (F_1,\ldots,F_D)$ denote a
  vector valued multilinear polynomial  and let $H_i = (T_{1-\epsilon}
  F_i)$ and $\mpl{H} = (H_1,\ldots,H_D)$.  Further let $\Inf_i(\mpl{H})
  \leq \tau$ and  $\Var[H_i] \leq 1$ for all $i$.

  If $\Psi : \R^D \rightarrow \R$ is a Lipschitz-continous function with
  Lipschitz constant $C_0$ (with respect to the $L_2$ norm).  Then,
  $$ \Big|\E\Big[\Psi(\mpl{H}(\mcl{L}^{\uglbl}))\Big] -
  \E\Big[\Psi(\mpl{H}(\erv{G}^\uglbl))\Big] \Big| \leq
  C_{D} \cdot C_0 \cdot \tau^{\epsilon/18K} = o_{\tau}(1) $$
  for some constant $C_D$ depending on $D$.
\end{theorem}

The local integral ensemble has an expected payoff equal to $c$.  The
local and global ensembles agree on the first two moments.  The
$(c,s)$-approximate polymorphism outputs a solution of value at least
$s$ when given the local integral ensemble as inputs.  Hence, by
invariance principle when given these global ensemble of random projections as input
to the polymorphism $\cP$, it
will end up outputting close to integral solutions of value $s$.

\begin{theorem}
  For all $\eta$ and CSP $\Lambda$, there exists $\tau,d > 0$ such that the
  following holds:
  Suppose $\cP$ is a $(c-\eta,s)$-approximate $(\tau,d)$-quasirandom
  polymorphism for \maxl.  Given an {\sc BasicSDP} solution with
  objective value at least $c$, the expected value of the assignment
  produced by $\round_{\cP}$ algorithm is at least $s - \eta$.
\end{theorem}
\begin{proof}
Let $\round_{\cP}(\vec V, \vec \mu)$ denote the
expected payoff of the ordering assignment by the rounding scheme
$\round_{\cP}$ on the SDP solution $(\vec V, \vec \mu)$ for the
$\Lambda$-CSP instance $\inst$.
Notice that the $\mrv{g}_i$ are nothing but samples of the global ensemble $\erv{G}$
associated with $\inst$. 
For each constraint $C \in \cC$, let $\cG_{C}$ denote the subset of random
variables from global gaussian ensemble that correspond to variables
in $\cV(\cC)$.  It will be useful to extend each payoff $C: [q]^k \to [-1,1]$
to the domain of fractional/distributional assignments $\ssimp_q$.
Specifically, for $z_1,\ldots,z_k \in \ssimp_q$, define $C(z_1,\ldots,z_k) \defeq \E_{x_i \sim
  z_i}[C(x_1,\ldots,x_k)]$ which corresponds to the expected payoff on
fixing each input $x_i$ independently from the corresponding
distribution $z_i \in \ssimp_q$.
By definition, the expected payoff is given
by
\begin{align}\label{eqn:roundvalue}
  \round_{\cP}(\vec V, \vec \mu) & = \E_{C \in \cC}[C\left(
  v_{1},\ldots,v_k \right)] \\
& = \E_{p \in \cP} \E_{C \in \cC} \E_{ {\erv{G}}_C^\uglbl }
  \Big[C\Big(\ssimptrunc\big(\mpl{H}({\mv{g}}^\uglbl_{v_1})\big),\ldots,\ssimptrunc\big(\mpl{H}({\mv{g}}^\uglbl_{v_k})\big)\Big)  \Big]
\end{align}
Fix a payoff
$C \in \cC$.  Let $\Psi_{C} : \R^{qk} \to \R$ be a Lipschitz continous
function defined as follows:
$$ \Psi_{C}(\mv{z}_{1},\mv{z}_2,\cdots, \mv{z}_k) =
C\Big(\ssimptrunc\big(\mv{z}_{1}\big),\ldots,\ssimptrunc\big({\mv{z}}_{k}\big)\Big)
\qquad \forall \mv{z}_1,\ldots \mv{z}_k \in \R^q\mper
$$
Hence, we can write
\begin{equation}
\round_{\cP}(\vec V, \vec \mu) = \E_{p\in \cP} \E_{C \in \cC} \E_{ {\erv{G}}_C^\uglbl }
    \Big[\Psi_C\Big(\mpl{H}\big({\mv{g}}^\uglbl_{v_1}\big),\ldots,\mpl{H}\big({\mv{g}}^\uglbl_{v_k}\big)\Big)
    \Big] \mper
\end{equation}
For a constraint $\cC$, there are $k$-different marginal distributions
in $\mu_{\cC}$.  Note that since $\cP$ is $(\tau,d)$-quasirandom, with probability at
least $ 1- k\sqrt{\tau}$ over the choice of $p \sim \cP$, we will
have
$$ \max_i \Inf_{i,\mu_j}^{\leq d}(p) \leq \sqrt{\tau} \mcom$$
for each of the marginals $\mu_j$ within $\mu_{\cC}$.
Since the objective value is always between $[-1,1]$, we get that
\begin{equation*}
\round_{\cP}(\vec V, \vec \mu)  \geq \E_{p\in \cP} \E_{C \in \cC} \E_{ {\erv{G}}_C^\uglbl }
    \Big[\Psi_C\Big(\mpl{H}\big({\mv{g}}^\uglbl_{v_1}\big),\ldots,\mpl{H}\big({\mv{g}}^\uglbl_{v_k}\big)\Big)
    |  \max_{i \in [R], j \in [k]}  \Inf_{i,\mu_j}^{<d}(p) \leq \sqrt{\tau} \Big] -
    k\sqrt{\tau}\\
\end{equation*}
Fix $d = (20/\e)\log(1/\tau)$.  For every such $p$, the polynomial $\mpl{H} = T_{1-\epsilon} p$, we can conclude that
\begin{align*}
 \max_{i \in [R], j \in [k]} \Inf_{i, \mu_j}(\mpl{H})
& \leq  \max_{i  \in [R], j \in [k]} \Inf_{i,\mu_j}^{<d}(T_{1-\e}p)+
(1-\e)^d\\
& \leq \max_{i \in [R], j \in [k]} \Inf_{i,\mu_j}^{<d}(p)+ (1-\epsilon)^d\\
&\leq \sqrt{\tau} + (1-\e)^d \leq 2\sqrt{\tau} \mper
\end{align*}
Applying the invariance principle (\prettyref{thm:invariance}) with the ensembles
$\mcl{L}_{C}$, $\erv{G}_{C}$, Lipschitz continous functional $\Psi$ and the vector
of $kq$ multilinear polynomials given by
$(\mpl{H},\mpl{H},\ldots,\mpl{H})$ where $\mpl{H} =
(H_1,\ldots,H_q)$, we get the required result:
\begin{align*}
  \round_{\cP}(\vec V, \vec \mu)  &\geq \E_{p\in \cP} \E_{C \in \cC}     \E_{\mcl{L}_{C}^\uglbl}
    \Big[\Psi_C\Big(\mpl{H}\big(\ell^\uglbl_{v_1}\big),\ldots,\mpl{H}\big(\ell^\uglbl_{v_k}\big)\Big)
    \Big] - o_{\tau}(1)
\end{align*}
where $o_{\tau}(1)$ is a function that tends to zero as $\tau \to 0$.
Since the local ensembles $\mcl{L}_C$ correspond to the local
  distributions $\mu_C$ over partial assignments $[q]^{\cV(C)}$ and
  $\mpl{H} = T_{1-\e}p$, we can rewrite the expression as,

\begin{align}
 \E_{p\in \cP} \E_{C \in \cC}\E_{\mcl{L}_{C}^\uglbl}
    \Big[\Psi_C\Big(\mpl{H}\big(\ell^\uglbl_{v_1}\big),\ldots,\mpl{H}\big(\ell^\uglbl_{v_k}\big)\Big)
    \Big]
& = \E_{p\in \cP} \E_{C \in \cC}\E_{(X_1,\ldots,X_k) \sim \mu_C^\uglbl}
    \Big[\Psi_C\Big(T_{1-\e}p(X_1),,\ldots,T_{1-\e}p(X_k)\Big)
    \Big] \mper\label{eq:012}
\end{align}
By definition of $\Psi_C$, $\Psi_C(\mv{z}_1,\ldots,\mv{z}_k) =
    C(\mv{z}_1,\ldots,\mv{z}_k)$  if $\forall i, \mv{z}_i \in \ssimp_q$.  Summarizing
    our calculation so far, we have
\begin{align*}
 \round_{\cP}(\vec V, \vec \mu) \geq  \E_{p\in \cP} \E_{C \in \cC}\E_{(X_1,\ldots,X_k) \sim \mu_C^\uglbl}
    \Big[\Psi_C\Big(T_{1-\e}p(X_1),,\ldots,T_{1-\e}p(X_k)\Big)
    \Big] - o_{\tau}(1) \mper
\end{align*}
The proof is complete with the following claim which we will prove in
the rest of the section.
\begin{claim} \label{claim:somename}
$$  \E_{p\in \cP} \E_{C \in \cC}\E_{(X_1,\ldots,X_k) \sim \mu_C^\uglbl}
    \Big[\Psi_C\Big(T_{1-\e}p(X_1),,\ldots,T_{1-\e}p(X_k)\Big)
    \Big] \geq s $$
\end{claim}

For each fixed $\eta$ and $\e = \eta/k$,  the error term $\kappa(\tau,d) \to 0$ as $d
\to \infty$ and $\tau \to 0$.  Therefore for a sufficiently large $d$ and
sufficiently small $\tau$, the error will be smaller than $\eta$.
\end{proof}
\begin{proof} (Proof of Claim \pref{claim:somename})
%
For $X \in [q]^{R}$, let $Z \sim_{1-\e} X$ denote the random variable
distributed over $[q]^R$ obtained by resampling each coordinate of $X$
from its underlying distribution with probability $\e$.  Notice that
$T_{1-\e} p(X_i) \in \ssimp_q$ is a fractional assignment.  To sample
a value $y \in [q]$ from the $T_{1-\e} p(X_i)$, we could instead sample
$Z_i \sim_{1-\e}  X_i$ and compute $p(Z_i)$.  In other words, $y \sim
T_{1-\e}p(X_i)$ is the same as $y = p(Z_i), Z_i \sim_{1-\e} X_i$.  Using the way
we defined the payoff function $C$ over fractional assignments, we get
\begin{align}
 &\E_{p\in \cP} \E_{C \in \cC}\E_{(X_1,\ldots,X_k) \sim \mu_C^\uglbl}
    \Big[C\Big(T_{1-\e}p(X_1),,\ldots,T_{1-\e}p(X_k)\Big)
    \Big] \nonumber\\
& \qquad \qquad = \E_{p \in \cP} \E_{C\in \cC} \E_{(X_1,\ldots,X_k)\sim
  \mu_C^{\uglbl}} \E_{Z_i \sim_{1-\e} X_i}\Big[ C\Big(
  p(Z_1),\ldots,p(Z_k)\Big)   \Big] \label{eq:123}
\end{align}
To lower bound the expression, construct an instance $\inst'$ that
consists of the constraints in $\inst$ but each over a disjoint set of
variables.  Consider the distribution $\Theta$ for assignments of $\inst'$
wherein the variables corresponding to each constraint $C \in \cC$ are
sampled independently from $\mu_C$.  Let $\Theta'$ be the $(1-\e)$-noisy
version of $\Theta$ in that it is obtained by sampling from $\Theta$ and
rerandomizing each coordinate with probability $\e$.

Notice that
\begin{align*}
\E_{Y \sim \Theta'}[\inst'(Y)] & = \E_{C \in \cC}\E_{x_1,\ldots,x_k \sim \mu_C}
\E_{z_1,\ldots,z_k \sim_{1-\e} x_1,\ldots,x_k}[C(x_1,\ldots,x_k)] \\
& \geq \E_{C \in \cC} \E_{x_1,\ldots,x_k \sim \mu_C} [C(x_1,\ldots,x_k)] - k\e\\
& \geq c -k\e
\end{align*}
where in the last inequality we used the fact that the objective value
of the SDP is at least $c$.  Since $\cP$ is a $(c-\eta,
s)$-approximate polymorphism and $\eta > k\e$, we can appeal to
\pref{lem:averagevalue} to conclude that,
\begin{align}
s & \leq \E_{p\in \cP}\E_{Y^{(1)},\ldots,Y^{(R)} \sim
    \Theta'}[\inst'(p(Y^{(1)},\ldots,p(Y^{(R)}))]  \nonumber \\
  &  \E_{p \in \cP} \E_{C\in \cC} \E_{(X_1,\ldots,X_k)\sim
  \mu_C^{\uglbl}} \E_{Z_i \sim_{1-\e} X_i}\Big[ C\Big( p(Z_1),\ldots,p(Z_k)\Big)
    \Big]  \label{eq:234}
\end{align}
The claim is immediate from \eqref{eq:123} and \eqref{eq:234}.
\end{proof}

\begin{proof} (Proof of \prettyref{thm:sdpgap})
  By definition of $s_{\Lambda}$, there exists an
  $(c,s_{\Lambda}(c))$-approximate polymorphism with degree $d$ influence less than
  $\tau$ for all $d, \tau$. Fix
  any particular instance $\inst$ of \prb{Max-$\Lambda$}.
  By applying the rounding scheme $\round_{\cP}$ on a
  sequence of polymorphisms with influence $\tau \to 0$ and
  $\eta \to 0$, we get that the SDP integrality gap on
  $\inst$ is at most $\lim_{\eta \to 0}s_{\Lambda}(c-\eta)$.
\end{proof}

\subsection{Approximate Polymorphisms and Dictatorship Tests}
\label{sec:approxpoly}


In this section, we will sketch the proof of
\prettyref{thm:ughardness}.
\begin{proof} (Proof of \pref{thm:ughardness})
  For each $\eta,\e > 0$, there exists $\tau_0, d_0$ such that there are no
  $(s_{\Lambda}(c + \e)+\eta,c + \e)$-approximate $(\tau_0,d_0)$-quasirandom
  polymorphisms for \maxl. By \pref{thm:ugpolymorphism}, this
  implies a unique games hardness to $(c,s_{\Lambda}(c+\e)c+\eta)$-approximate $\maxl$.
  By making $\e \to 0$ and $\eta \to 0$, we get a unique games based hardness
  for  $(c, \lim_{\e \to 0} s_{\Lambda}(c+\e))$-approximating \maxl.

  By Observation \prettyref{obs:approxpoly}, we have
  $$  \lim_{\e \to 0} s_{\Lambda}(c+\e)  \leq s_{\Lambda}(c) \mper$$
  Hence, the conclusion of \pref{thm:ughardness} follows.
\end{proof}

\begin{theorem} \label{thm:ugpolymorphism}
  Fix constants $\tau,c$ and $d \in \N$.
  Suppose \maxl does not admit a $(c,s)$-approximate
  $(\tau,d)$-quasirandom polymorphism then for all $\eta > 0$, it is unique games hard to
  $s+\eta$-approximate instances of \maxl on
  instances with value at least $c-\eta$.
\end{theorem}

\begin{proof}
  %

  Fix an integer $R \in \N$.  For all $c,s$, the set of $(c,s)$-approximate polymorphisms $\cP$ of
  arity $R$ form a convex set.  Convex combination of two approximate polymorphisms
  $\cP_1$ and $\cP_2$ is constructed by taking a convex combination of the
  underlying distributions over operations.
  For every $(c ,s)$-approximate polymorphism $\cP$, there
  exists a probability distribution $\mu \in \ssimp_q$ such that,
  $$ \E_{p \in \cP}\left[ \max_{i \in [R]} \Inf^{<d}_{i,\mu}(p)\right] \geq
  \tau \mper$$
  By min-max theorem, there exists a distribution $\Phi$ over the simplex
  $\ssimp_q$ such that for every $(c,s)$-approximate
  polymorphism $\cP$,
  \begin{align} \label{eq:influencelb1}
    \E_{p\in \cP} \E_{\mu \sim \Phi}\left[ \max_{i \in [R]} \Inf_{i,\mu}^{<d}(p)
    \right] \geq \tau
  \end{align}
  Define the family of functions $\cF^{\Phi,R}_{\tau,d}$ as follows,
  $$ \cF^{\Phi,R}_{\tau,d} \defeq \left\{ p:[q]^R \to [q] \suchthat \E_{\mu\sim
      \Phi} \left[  \max_i \Inf^{<d}_{i,\mu} (p) \right] \leq \tau\right\} \mper$$

  By \eqref{eq:influencelb1}, there does not exist $(c, s)$-approximate polymorphisms
  supported in $\cF^{\Phi,R}_{\tau,d}$.  By the connection between approximate
  polymorphisms and dictatorship tests outlined in
  \pref{thm:poly-to-dict-equivalence}, this implies that there exists $(c,s)$-dictatorship tests against the family
  $\cF_{\tau,d}^{\Phi,R}$ for each $R \in \N$.
  Let $\inst^{\Phi,R}_{\tau,d}$ denote the $(c,s)$-dictatorship test against the family $\cF_{\tau,d}^{\Phi,R}$.

  These dictatorship gadgets can be utilized towards carrying out the unique games
  based hardness reduction.  To this end, we will need to show
  a stronger soundness guarantee for the dictatorship test.
  Specifically, we will have to show a soundness guarantee against
  functions that output distributions over $[q]$ (points in $\ssimp_q$)
  instead of elements over $[q]$.  The details are described below.

  First, we will extend the objective function associated with an instance
  $\inst$ of \maxl from $[q]$-valued assignments to $\ssimp_q$-valued
  assignments.
  Given a fractional assignment $z : \cV \to \ssimp_q$, let $z^{\times}$ denote the product
  distribution over $[q]^{\cV}$ whose marginals given by $z$.  Define
  $\inst: \ssimp_q^\cV \to\R$ as follows,
  $$\inst(z) \defeq \E_{x \sim z^{\times}}[\inst(x)]$$
  Notice that the definitions of influences and low-degree influences
  extend naturally to $\ssimp_q$-valued functions.  In fact, even for
  $[q]$-valued functions these notions were defined by expressing them as
  $\ssimp_q$-valued functions.
  We will show that quasirandom fractional assignments have no larger
  objective value than quasirandom $[q]$-valued assignments.
  Specifically, we will show the following.
  \begin{lemma} \label{thm:rounding1}
    For every $\e >0 $, the following holds for all sufficiently large $R \in \N$.
    For every function $p: [q]^R \to\ssimp_q$ such that $\E_{\mu \sim \Phi}\left[ \max_i
      \Inf_{i,\mu}^{<d}(p) \right] \leq \tau/5$,
    $$ \inst^{\Phi,R}_{\tau,d}(p) \leq s + \e$$
  \end{lemma}
  \begin{proof}
    The idea is to create a rounded
    operation $r : [q]^R \to [q]$ by setting for each $x
    \in [q]^R$,
    $$r(x) = \text{ random sample from distribution } p(x)$$
    For notational convenience, we will drop the subscripts and
    superscripts and write $\inst$ for $\inst^{\Phi,R}_{\tau,d}$.
    It is easy to see that by definition,
    $$
    \E_{r}\left[\inst(r(X^{(1)},\ldots,X^{(R)}))\right]
    = \inst(p(X^{(1)},\ldots,X^{(R)}))$$
    The technical core of the argument shows that if $p$ is
    {\it quasirandom} then the rounded function $r$ is {\it quasirandom}
    with high probability.  This claim is formally stated in
    \pref{lem:quasirandom-concentration}.  By
    \pref{lem:quasirandom-concentration}, for all sufficiently large $R
    \in \N$,
    $$ \Pr_{r} \left\{  \E_{\mu \sim \Phi}[ \max_i \Inf_{i,\mu}^{<d}(r)] < \tau \right\} \geq 1-\e/R\tau $$
    Call a rounded function $r: [q]^R \to [q]$ to be {\it quasirandom}, if
    $\E_{\mu \sim \Phi}[\max_i \Inf_{i,\mu}^{<d}(r)]< \tau$.
    $$ \E_{r \sim p}[ \inst(r) |r \text{ is quasirandom }] \geq \E_{r \sim
      p}[\inst(r)] -\Pr_r[r \text{ is not quasirandom}] \geq \inst(p) - \e/R\tau$$
    For a {\it quasirandom} function $r:[q]^R \to [q]$, the dictatorship
    test $\inst$ satisfies $\inst(r) \leq s$.  This implies that
    $\inst(p) \leq s+ \e$.
  \end{proof}

  Now, we will outline the details of the unique games hardness
  reduction.  Given a unique games instance $\Gamma = (\cA,\cB, E,\Pi,[R])$,
  the reduction produces an instance $\inst_\Gamma$ of \maxl.  The
  variables of $\inst_{\Gamma}$ are $\cB \times [q]^R$.  The constraints of
  $\inst_\Gamma$ can be sampled using the following procedure.

  \begin{itemize}
  \item Sample $a \in \cA$ and neighbors $b_1,\ldots, b_k \in \cB$ independently
    at random.
  \item Sample a constraint $c$ from the dictatorship test
    $\inst^{\Phi,R}_{\tau,d}$.  Suppose the constraint is on{}
    $x^{(1)},x^{(2)},\ldots,x^{(k)} \in [q]^R$
  \item Introduce the constraint $C$ on $\{(b_i, \pi_{ab_i}\circ x^{(i)})\}_{i=1}^k$.
  \end{itemize}
  Given an assignment $L: \cB \times [q]^R \to [q]$ its value is given by,

  \begin{equation} \label{eq:uginstancevalue}
    \inst_\Gamma(L) = \E_{a \in \cA} \E_{b_1,\ldots, b_k \in N(a)} \E_{C \in
      \inst^{\Phi,R}_{\tau,d}} \left[  C(L(b_i,\pi_{ab_i}\circ x^{(i)})) \right]
  \end{equation}

  \paragraph{Completeness}  Suppose $\cL: \cA \cup \cB \to [R]$ is a labelling
  satisfying $(1-\delta)$-fraction of constraints.  Consider the labelling
  $L: \cB \times [q]^R \to [q]$ given by $L(b, x) = x_{\cL(b)}$.
  Over the choice of $a,b_1,\ldots,b_k$, with probability at least $(1-\delta k )$,
  the labelling $\cL$ satisfies each of the $k$ edges $\{(a,b_i) | i \in
  [k]\}$.  In this case, the objective value is at least the completeness $c$
  of the dictatorship test $\inst^{\Phi,R}_{\tau,d}$.  With the remaining
  probability, the objective value is at least $-1$.  This implies that
  the labelling $L$ has an objective value $\geq c-2k\delta$.

  \paragraph{Soundness}

  Suppose $L : \cB \times [q]^R \to [q]$ is a labelling with an objective value
  $s + \eta$.  For each $b \in \cB$, set $L_b : [q]^R  \to [q]$ as $L_b(x)
  \defeq e_{L(b,x)}$ where $e_{L(b,x)} \in \ssimp_q$ corresponds to a basis
  vector along direction $L(b,x)$.  For each $a \in \cA$, define the
  fractional assignment $L_a: [q]^R \to \ssimp_q$ as $L_a(x)
  \defeq \E_{b \in N(a)}[L_b(\pi_{ab}\circ x)]$.

  The objective value in \eqref{eq:uginstancevalue} can be written as,
  $$ \inst_{\Gamma}(L) = \E_{a \in \cA}\left[ \inst^{\Phi,R}_{\tau,d}(L_a)
  \right] \mper$$

  Decode an assignment $\ell: \cA \cup \cB \to [R]$ as follows.  Sample a
  distribution $\mu \sim \Phi$.  For each $a \in \cA$, define the label set
  $$ T_a = \{ i | \Inf_{i,\mu}^{<d}(L_a) \geq \tau/10\}$$
  and for each $b \in \cB$, define the label set
  $$ T_b = \{ j | \Inf_{j,\mu}^{<d}(L_b) \geq \tau/20\} \mper$$
  Since the sum of degree $d$ influences is at most $d$, we have $|T_a|
  \leq 10d/\tau$ and $|T_b| \le 20d/\tau$ for all $a \in \cA$ and $b \in \cB$.

  For each $a \in \cA$, assign $\ell(a)$ to be a random label from $T_a$ if it
  is non-empty,  else assign an arbitrary label from $[R]$.  Similarly,
  for each $b \in \cB$, set $\ell(b)$ to be a random label from $T_b$ if it
  is non-empty, else assign an arbitrary label from $[R]$.

  If the objective value of the assignment $L$ is more than $s + \eta$,
  then for at least $\eta/2$ fraction of $a \in \cA$, we have
  $\inst^{\Phi,R}_{\tau,d}(L_a) > s + \eta/2$.
  Call such a vertex $a \in \cA$ to be {\it good}.
  Fix a good vertex $a \in \cA$.
  For every {\it good} vertex, by the soundness of the dictatorship test we will have,
  $$ \E_{\mu \sim \Phi} [\max_{i} \Inf_{i,\mu}^{<d}(L_a)] \geq \tau \mper$$
  This implies that over the choice of $\mu \sim \Phi$, with probability at
  least $\tau/2$, we will have $\max_i \Inf_{i,\mu}^{<d}(L_a) \geq \tau/2$
  which implies that $L_a$ is non-empty.
  Suppose $i \in L_a$.  Using the fact that $L_a = \E_{ b \in N(a)} \pi_{ab} \circ L_{b}$ and convexity of
  influences we have,
  $$  \E_b [ \Inf_{\pi_{ab}(i),\mu}^{<d}(L_b)] \geq \tau/2 \mper$$
  This implies that for at least $\tau/4$-fraction of neighbors $b$, we
  will have $\Inf_{\pi_{ab}(i),\mu}(L_b) \geq \tau/4$, i.e., $\pi_{ab}(i)
  \in T_b$.  For every such neighbor $b$, the labelling $\ell$ satisfies the
  edge $(a,b)$ with probability at least $\frac{1}{|T_a|} \cdot \frac{1}{|T_b|} \geq \tau^2/200d^2$.

  Hence, the fraction of edges satisfied by the labelling $\ell$ is at
  least $(\eta/2) \cdot (\tau/2) \cdot (\tau/4) \cdot(\tau^2/200d^2)$ in expectation.  By fixing
  an alphabet size $R$ large enough, the soundness $\delta$
of the unique games
  instance can be made smaller than this fraction.

\end{proof}


\subsection{Quasirandomness under Randomized Rounding}

This section is devoted to showing the following lemma which was used
in the proof of \pref{thm:ugpolymorphism}.

\begin{lemma}\label{lem:quasirandom-concentration}
  For every $\tau,d, \e$, the following holds for all large enough $R \in \N$:
  Suppose $p:[q]^R \to \ssimp_q$ and a distribution $\Phi$ over $\ssimp_q$ is
  such that,
  $$ \E_{\mu \sim \Phi}[ \max_i \Inf_{i,\mu}^{<d}(p)] < \tau $$
  then if $r:[q]^R \to [q]$ is sampled by setting each $r(x) \in [q]$
  independently from the distribution $p(x)$,
  $$ \Pr_{r} \left\{  \E_{\mu \sim \Phi}[ \max_i \Inf_{i,\mu}^{<d}(r)] < 5\tau \right\} \geq 1-\e/R\tau $$
\end{lemma}

To this end, we first show a few simple facts concerning concentration
under random sampling.

\paragraph{Random Sampling}

\begin{lemma} \label{lem:concentration}
  Given $p: [q]^R \to \ssimp_q$ let $r :[q]^R \to [q]$ denote be a function
  sampled by setting $r(x) \in [q]$ from the distribution $p(x) \in
  \ssimp_q$.  Fix a probability distribution $\mu \in \ssimp_q$ and a
  corresponding orthonormal basis $\{ \chi_\sigma \}_{\sigma \in [q]^R}$ for $L^2([q]^R,\mu^R)$.  For every
  multi-index $\sigma \in [q]^R$,
  $$ \E_R [\hat{r}_\sigma] = \hat{p}_\sigma \mcom$$
  and
  $$\Pr_r[|\hat{r}_\sigma - \hat{p}_\sigma| > \beta] \leq 2 \exp\left(-\frac{\beta^2}{\left(\max_{i \in
          [q]} \mu(i)\right)^R} \right) \mper$$
\end{lemma}
\begin{proof}
  Fix a multi-index $\sigma \in [q]^R$.  Let $\chi_\sigma : [q]^R \to \R$ denote the
  corresponding basis function for $L^2([q]^R, \mu^R)$.
  Let $\mu^R(x)$ denote the probability of $x \in [q]^R$ under the product
  distribution $\mu^R$.
  By definition,
  $$\hat{r}_\sigma = \sum_{x \in [q]^n}\mu(x)\chi_\sigma(x)r(x) $$
  where $\{r(x)\}_{x \in [q]^R}$ are independent random variables whose
  expectation is given by $\E r(x) = p(x)$.  For every $x$, the random
  vector $\mu(x) \chi_\sigma(x)r(x)$ has entries bounded in
  $[0,\mu(x)\chi_\sigma(x)]$.  Now we appeal to Hoeffding's inequality stated below.
\begin{lemma}\label{lem:Hoeffding}(Hoeffding's inequality)
  Given independent real valued random variables $X_1,\ldots,X_n$ such that
  $X_i \in [b_i, a_i]$, we have
  $$\Pr\left[ \abs{\sum_i X_i - \E[\sum_i X_i]} \geq t \right] \leq 2 \exp\left( -\frac{t^2}{\sum_i (b_i-a_i)^2} \right) $$
\end{lemma}
By Hoeffding's inequality, we will have
  $$\Pr_r[|\hat{r}_\sigma - \hat{p}_\sigma| > \beta] \leq 2 \exp\left(-\frac{\beta^2}{\sum_x \mu(x)^2\chi_\sigma(x)^2}\right) \mper$$
  The proof is complete once we observe that $\sum_x \mu(x)^2 \chi_\sigma(x)^2 \leq \max_{x
  } \mu(x) \cdot \sum_{x} \mu(x)\chi_\sigma^2(x) = \max_x \mu(x) \le
  (\max_i \mu(i))^R$.

\end{proof}

\begin{lemma} \label{lem:boundonbias}
  Given a function $r:[q]^R \to \R^d$ and a probability distribution $\mu
  \in \ssimp_q$, for all $i \in [R]$ and $D \in \N$,
  $$\Inf^{<D}_{i,\mu}(r) \leq 8 \left( 1-\max_{i \in [q]} \mu(i) \right) \cdot
  \max_x \norm{r(x)}^2$$
\end{lemma}

\begin{proof}
  First, note that for each $D \in \N$ we have $\Inf_{i,\mu}^{<D}(r) \leq
  \Inf_{i,\mu}(r)$.  Moreover, the influence of the $i^{th}$ coordinate
  can be written as,
  \begin{align} \label{eq:infvariance}
    \Inf_{i,\mu}(r) = \E_{x_{[n]\backslash i} \sim \mu^{R-1}} \left[
    \Var_{x_i \sim \mu} r(x) \right] \mper
  \end{align}
  For each fixing of $x_{[n]\backslash i} \in [q]^{R-1}$, we will bound the
  variance of $r(x)$ over the choice of $x_i$ as follows,
  \begin{align*}
    \Var_{x_i \sim \mu}[r(x)] & = \E_{z,z' \sim \mu}\left[ \norm{r(x_{[n]\backslash i},z) -
                          r(x_{[n]\backslash i},z')}^2 \right] \\
                        & \leq 4 \max_{x \in [q]^R} \norm{r(x)}^2 \cdot \Pr[z \neq z'] \\
                        & \leq 8 \left( \max_{x \in [q]^R} \norm{r(x)}^2\right) \cdot \left( 1 - \max_i \mu(i)\right) \mper
  \end{align*}
  The result follows by using the above bound in \eqref{eq:infvariance}.
\end{proof}

\begin{lemma} \label{lem:influenceconcentration}
  For every $\tau,d,\e$, the following holds for all large enough $R \in \N$:
  Given $p: [q]^R \to \ssimp_q$ let $r :[q]^R \to [q]$ denote a function
  sampled by setting $r(x) \in [q]$ from the distribution $p(x) \in
  \ssimp_q$.  For every distribution $\mu \in \ssimp_q$ and $i \in [R]$,
  $$ \Pr[ \Inf_{i,\mu}^{<d}(r) > 2\Inf_{i,\mu}^{<d}(p) + \tau] < \nicefrac{\e}{R^2}$$
\end{lemma}
\begin{proof}
  By \pref{lem:boundonbias}, if $\max_{j \in [q]} \mu(j) > 1-
  \tau/16$ then we will have $\Inf_{i,\mu}^{<d}(r) < \tau$ for all
  functions $r:[q]^R\to[q]$.  Hence the statement trivially holds if
  $\max_{j \in [q]} \mu(j) > 1-\tau/16$.

  Suppose $\max_{j\in [q]}\mu(j) \leq 1-\tau/8$.  In this case, for every
  multi-index $\sigma \in [q]^R$, \pref{lem:concentration} implies that
  $$ \Pr[ |\hat{r}_{\sigma} - \hat{p}_{\sigma}| \geq \beta]  \leq 2 \exp\left(
    - \frac{\beta^2}{(1-\tau/8)^R}\right) \mper$$
  By a union bound over all $\sigma$ with $|\sigma| < d$, we will have
  $$ \Pr[ |\hat{r}_\sigma - \hat{p}_{\sigma}| \geq \beta \quad \forall |\sigma| <
  d] \leq 2(qR)^d  \exp\left(- \frac{\beta^2}{(1-\tau/8)^R}\right) \mper$$
  Fix $\beta = (qR)^{-d}$.  Fix $R$ large enough so that $1/(qR)^d < \tau$ and
  the above probability bound is smaller than $\e/R^2$.
  If $|\hat{r}_\sigma - \hat{p}_\sigma| \leq \beta$ for all $\sigma
  $ with $|\sigma| < d$, then
  $$ \Inf_{i,\mu}(r) \leq  \sum_{\sigma_i \neq 0, |\sigma| <d}
  (\hat{p}_{\sigma} + \beta)^2  \leq  \sum_{\sigma_i \neq 0, |\sigma| < d}
  2\hat{p}^2_\sigma + 2\beta^2 \leq 2\Inf_{i,\mu}^{<d}(p) + 2(qR)^{-2d}(qR)^d \leq 2
  \Inf_{i,\mu}^{<d}(p) + \tau $$

\end{proof}

Now we are ready to prove \pref{lem:quasirandom-concentration}.
\begin{proof} (Proof of \pref{lem:quasirandom-concentration})
  Fix a distribution $\mu \sim \Phi$.  Call a probability distribution $\mu \sim \Phi$ to be {\it bad} if
  $$ \max_i \Inf_{i,\mu}^{<d}(r) > 2 \max_i \Inf_{i,\mu}^{<d}(p) +
  \tau $$
  For each distribution $\mu \in \ssimp_q$, by \pref{lem:influenceconcentration} and a union bound over $i \in [R]$
  $$ \Pr_r[\mu \text{ is bad}]  \leq \e/R$$
  By Markov's bound,
  \begin{equation} \label{eq:inf201}
    \Pr_r\left\{ \E_{\mu \sim \Phi}\Ind[\mu \text{ is bad}] > \tau \right\} < \e/R\tau
  \end{equation}
  Let us suppose $ \E_{\mu \sim \Phi}[\Ind[\mu \text{ is bad}]] < \tau$.
  In this case, we conclude that,
  \begin{equation}\label{eq:inf202}
    \E_{\mu \sim \Phi}[\max_i \Inf_{i,\mu}^{<d}(r)] \leq \left( 2 \E_{\mu \sim
        \Phi}[\max_i \Inf_{i,\mu}^{<d}(p)] + \tau \right)+ \tau \cdot 2
  \end{equation}
  where we used the fact that for every distribution $\mu \in \ssimp_q$ and every function $r:
  [q]^R \to [q]$ we have $\max_i \Inf_{i,\mu}^{<d}(r) < 2$.  The claim
  follows from \eqref{eq:inf201} and \eqref{eq:inf202}. \end{proof}

\bibliographystyle{amsalpha}
\bibliography{bib/papers}

\appendix

\end{document}